\documentclass{fac}
\usepackage{graphicx}
\usepackage{amssymb}
\usepackage{amsmath}
\ifprodtf \else \usepackage{latexsym}\fi

\newcommand\black{\ensuremath{\blacktriangleright}}
\newcommand\white{\ensuremath{\vartriangleright}}

\newif\ifamsfontsloaded
\ifprodtf
  \newcommand\whbl{\white\kern-.1em--\kern-.1em\black}
  \newcommand\blwh{\black\kern-.1em--\kern-.1em\white}
  \newcommand\blbl{\black\kern-.1em--\kern-.1em\black}
  \newcommand\whwh{\white\kern-.1em--\kern-.1em\white}
  \amsfontsloadedtrue
\else
  \checkfont{msam10}
  \iffontfound
    \IfFileExists{amssymb.sty}
      {\usepackage{amssymb}\amsfontsloadedtrue
       \newcommand\whbl{\white\kern-.125em--\kern-.125em\black}%
       \newcommand\blwh{\black\kern-.125em--\kern-.125em\white}%
       \newcommand\blbl{\black\kern-.125em--\kern-.125em\black}%
       \newcommand\whwh{\white\kern-.125em--\kern-.125em\white}}
      {}
  \fi
\fi

\newtheorem{theorem}{Theorem}[section]

\title[Formal Model of Web Service Composition: An Actor-Based Approach to Unifying Orchestration and Choreography]
      {Formal Model of Web Service Composition: An Actor-Based Approach to Unifying Orchestration and Choreography}

\author[Yong Wang]
    {Yong Wang\\
     College of Computer Science and Technology,\\
     Beijing University of Technology, Beijing, China\\
     }

\correspond{Yong Wang, Pingleyuan 100, Chaoyang District, Beijing, China.
            e-mail: wangy@bjut.edu.cn}

\pubyear{2011}
\pagerange{\pageref{firstpage}--\pageref{lastpage}}

\begin{document}
\label{firstpage}

\makecorrespond

\maketitle

\begin{abstract}
Web Service Composition creates new composite Web Services from the collection of existing ones to be composed further and embodies the added values and potential usages of Web Services. Web Service Composition includes two aspects: Web Service orchestration denoting a workflow-like composition pattern and Web Service choreography which represents an aggregate composition pattern. There were only a few works which give orchestration and choreography a relationship. In this paper, we introduce an architecture of Web Service Composition runtime which establishes a natural relationship between orchestration and choreography through a deep analysis of the two ones. Then we use an actor-based approach to design a language called AB-WSCL to support such an architecture. To give AB-WSCL a firmly theoretic foundation, we establish the formal semantics of AB-WSCL based on concurrent rewriting theory for actors. Conclusions that well defined relationships exist among the components of AB-WSCL using a notation of Compositionality is drawn based on semantics analysis. Our works can be bases of a modeling language, simulation tools, verification tools of Web Service Composition at design time, and also a Web Service Composition runtime with correctness analysis support itself.
\end{abstract}

\begin{keywords}
Web Services; Web Service Orchestration; Web Service Choreography; Actor Systems; Compositionality; Rewriting Semantics; Interaction Semantics
\end{keywords}

\section{Introduction}

Web Service (WS) is a quite new distributed software component which emerged about ten years ago to utilize the most widely-used Internet application protocol--HTTP as its base transport protocol. As a component, a WS has the similar ingredients as other ones, such as DCOM, EJB, CORBA, and so on. That is, a WS uses HTTP-based SOAP\cite{SOAP} as its transport protocol, WSDL\cite{WSDL} as its interface description language and UDDI\cite{UDDI} as its name and directory service.

WS Composition creates new composite WSs using different composition patterns from the collection of existing WSs. Because of advantages of WS to solve cross-organizational application integrations, two composition patterns\cite{WSOandWSC1} are dominant. One is called Web Service Orchestration (WSO), which uses a workflow-like composition pattern to orchestrate business activities (implemented as WS Operations) and models a cross-organizational business processes or other kind of processes. The other is called Web Service Choreography (WSC) which has an aggregate composition pattern to capture the external interaction behaviors of WSs and acts as a contract or a protocol among WSs.

We now take a simple example of buying books from a book store to illustrate some concepts of WSComposition. Though this example is quite simple and only includes the sequence control flow (that is, each business activity in a business process is executed in sequence), it is enough to explain the concepts and ideas of this paper and avoids unnecessary complexity without loss of generality. We use this example throughout this paper. The requirements of this example are as Fig.\ref{Fig.5} shows.

\begin{figure}
  \centering
  \includegraphics{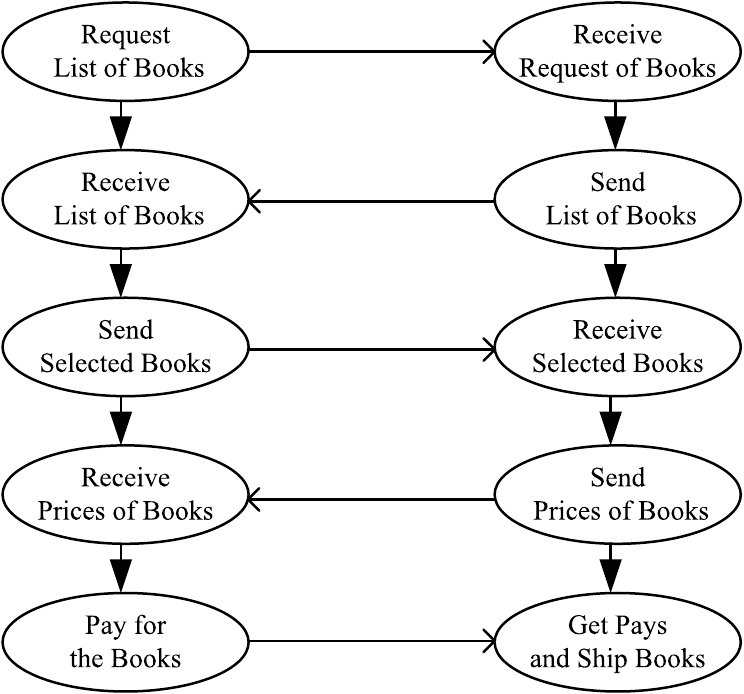}
  \caption{Requirements of an Example.}
  \label{Fig.5}
\end{figure}

A customer buys books from a book store through a user agent. In this example, we ignore interactions between the customer and the user agent, and focus on those between the user agent and the book store. Either user agent or book store has business processes to interact with each other.

We give the process of user agent as follows. The process of book store can be gotten from that of user agent as contrasts.

\begin{enumerate}
  \item The user agent requests a list of all books to the book store.
  \item It gets the book list from the book store.
  \item It selects the books by the customer and sends the list of selected books to the book store.
  \item It receives the prices of selected books from the book store.
  \item It accepts the prices and pays for the selected book to the book store. Then the process terminates.
\end{enumerate}

Since the business activities, such as the book store accepting request for a list of books from the user agent, are implemented as WSs (exactly WS operations), such buyer agent and book store business processes are called WSOs. These WSOs are published as WSs called their interface WSs for interacting among each other. The interaction behaviors among WSs described by some contracts or protocols are called WSCs.

There are many efforts for WS Composition, including its specifications, design methods and verifications, simulations, and runtime supports. Different methods and tools are used in WS Composition research, such as XML-based WSO description specifications and WSC description specifications, formal verification techniques based on Process Algebra and Petri-Net, and runtime implementations using programming languages. Some of these works mainly focus on WSO, others mainly on WSC, and also a few works attempt to establish a relationship between WSO and WSC.

Can a WS interact with another one? And also, can a WSO interact with another one via their interfaces? Is the definition of a WSC compatible with its partner WSs or partner WSOs? To solve these problems, a correct relationship between WSO and WSC must be established. A WS Composition system combining WSO and WSC, with a natural relationship between the two ones, is an attractive direction\cite{WSOandWSC1}. In such a system, not only WSO is focused like the famous WSO engine -- ActiveBPEL\cite{ActiveBPEL}, some monitoring tools and modeling tools, but also WSC is included. And in a systematic viewpoint, WS, WSO and WSC are organized with a natural relationship under the whole environment of cross-organizational business integration. More importantly, such a system should have firmly theoretic foundation with formal semantics support.

In this paper, we try to make such a system to base on actor systems theory\cite{Actor1}\cite{Actors}\cite{ConcurrentObject}. The contributions of our works are: (1)Through analysis of WS, WSO and WSC, we design a WS Composition runtime architecture with which a natural relationship among the three ones is established. (2)We introduce an actor-based language called AB-WSCL to support such an architecture. (3)More importantly, the rewriting semantics of AB-WSCL are given. Especially, through interaction semantics between WSO and WS, WS and WS, WSO and WSO, the compositionality concept is used to give finely formal relations among these system components. (4)To show the values of our works, mappings between AB-WSCL and XML-based WS specifications are illustrated. Based on the mappings, our works can be used as not only a validation tools in design time, but also a runtime system itself.

The organization of this paper is following: In section 2, related works are analyzed. We design an architecture of WS Composition runtime based on analysis of WSO, WSC and their relationship in section 3. Based on the runtime architecture, AB-WSCL is introduced in section 4. Formal semantics based on actor rewriting theory\cite{ActorRewriting}\cite{ActorRewriting1}\cite{ActorRewriting2} are given in section 5. And finally, we conclude our current works and give future directions in section 6.

\section{Related Works}

We introduce the related works about WSO, WSC, and their relations. Especially, we detail the works about formalization of WS Composition.

\subsection{WS Composition and Cross-Organizational Workflow}

Workflow\cite{WFMCRM}\cite{WFMCPD} is also a solution to cross-organizational business process integration. Workflow in such situation is called cross-organizational workflow\cite{InterorgaWorkflowPNet}\cite{InterorgaWorkfowPView1}\cite{InterorgaWorkfowPView2}. But the implementations of activities in a cross-organizational workflow are not limited to WSs.

Aalst uses Petri Net to model a workflow called WF-Net\cite{WF-Net} and connects two WF-Nets within different organizations into a newly global WF-Net to model integration of two workflows. But, a partner workflow is located within an organization, that is, the details of an inner workflow are hidden to the external world, so a global view of entire connection with two detailed inner workflows usually can not be gotten.

This leads to emergence of the so-called process view\cite{InterorgaWorkfowPView2} and integration of two inner process can be implemented based on process views\cite{InterorgaWorkfowPView1}. A process view is an observable version of an inner process from outside and serves as the interface of an inner process.

Research efforts on cross-organizational workflow give a great reference to WS Composition, especially WSO. WSOs aimed at cross-organizational business process integrations and used a workflow-like pattern to orchestrate WSs are surely some kinds cross-organizational workflows within which have activities implemented by WSs outside.

\subsection{The Industry Standardization of WS Composition}

In the viewpoint of W3C, a WS itself is an interface or a wrapper of an application inside the boundary of an organization that has a willing to interact with applications outside. That is, a W3C WS has not an independent programming model like other component models and has not needs of containing local states for local computations. Indeed, there are different sounds of developing WS to be a full sense component, Such as OGSI\cite{OGSI}. Incompatibility between W3C WS and OGSI-like WS leads to WSRF\cite{WSRF} as a compromised solution which reserves the W3C WS and develops a notion of WS Resource to model states. Since interactions among WSs are eventually driven by their inner applications at runtime, there is no any necessary to maintain states for a WS. That is, a WS just \emph{deliver}\footnote[1]{without any processing} messages, including incoming messages and outgoing messages.

A WSO orchestrates WS operations into a workflow to deal with modeling of business process. The industry tries to define a uniform WSO description language specification. The first one is WSFL\cite{WSFL} which is a flow language based on directed acyclic graph (DAG). Note that the flow model of WSFL is used to define a WSO, WSFL also provides a so-called global model to aggregate WSs. That is, global model gives a prototype to define a WSC. Another one emerged at almost the same time is XLANG\cite{XLANG} which uses a structural model and provides programming constructs to define a WSO. WSFL and XLANG are converged to WS-BPEL\cite{WS-BPEL} which adopts not only a DAG model, but also a structural model.

As a runtime mechanism, WSRF-style WS can serve as an interface of a WSO. But, a WSO, even a WS Resource, is surely different to a stateless function for share. That is, from a view of an observer outside, the interfaces\footnote[1]{exactly WSDL Operations of an interface WS of a WSO} of a WSO surely are stateful and must be invoked under the constraints of the WSO. Note that, different observers may have different views to the same WSO. Such observable views of a WSO are called abstract processes in WS-BPEL\cite{WS-BPEL} and are called WSC Interfaces in WSCI\cite{WSCI}. But these specifications do not explain their functions and even the relations to the original WSO. Though a WSRF-style WS can act as the interfaces of a WSO at runtime, but it has few uses at design time. Since observable views have more information about the inner WSO, they can be used to do more things at design time, such as generation of a WSC and generation of stub codes being more expressive than those of WSDL.

At the beginning, main efforts by the industry are done for WSO specifications though global model of WSFL can be used to define a WSC. People realize that the definition of interaction behaviors among WSs needs a separate specification. This leads to the occurrence of WSCI\cite{WSCI} and late WS-CDL\cite{WS-CDL}. Both WSCI and WS-CDL capture the interactions among WSs, including partner roles, messages exchanged definitions, linked WS operations and their executing sequence. A WSC acts as a contract or a protocol among interacting WSs. In a viewpoint of business integration requirements, a WSC serves as a business contract to constrain the rights and obligations of business partners. And from a view of WS technology, we can treat a WSC as a communication protocol which coordinates the interaction behaviors of involved WSs.

\subsection{Formalization of WS Composition}

The research works on formalization of WS Composition are almost numerous, we only introduce some representatives.

\subsubsection{Formalization for WSO and Its Interface}

Formalization for WSO includes formalization for WSO, its interface, and the relationship between WSO and its interface.

The first issue on formalization for WSO is formalizing control flows in a WSO. The control flows may be expressed in graphic model or structural model. Formalization for control flows may be based on different formal tools. By establishing the maps between control flows in a WSO and constructs in one formal model, control flows can be translated into expressions in the formal model. Then the properties, such as liveness and safety, can be verified.

The observable view of a WSO also includes control flows and inner transaction process mechanisms, so formalization for an observable view has the same issues as that of a WSO. Since the interface of a WSO includes the observable view of the WSO which is more expressive than the WSDL-expressed WS. That is, formalization for control flows and inner technical mechanisms may also be faced with.

A WSO interacts with outside via its interface, but, can the interface act in behalf of the inner WSO? That is, there must be a correct relationship between a WSO and its interface. Since observable view of a WSO is an observable version of the inner WSO, formalization for relationship between a WSO and its interface can be based on some kind of equivalences, such as trace equivalence, bisimulation equivalence, and so on. Under such concepts of equivalence, a correct relationship between a WSO and its interface can be established. That is, an interface is \emph{same} as its inner WSO modulo such equivalence. Automatic generation of an interface from an inner WSO and also refinement of an inner WSO with an interface are also contents of a formal model.

Petri Net has good traditions in formalizing WSO and can be casted back to Aalst's WF-Net\cite{WF-Net}. \cite{FormBPELPNet} tries to give BPEL a formal semantics based on Petri Net. The works cover the standard behaviors of BPEL, including atomic activities, control flows in term of structural activities, exception processing, event processing and LRBT mechanism.\cite{FormBPELPNet2} also transforms a WSO described by WS-BPEL into a so-called Service Workflow Net (SWN), which is a kind of colored Petri Net. And then it analyzes the compatibility of two services.

Process Algebra is also an influential tool in formalizing WS Composition. \cite{FormBPELPA} uses a Process Algebra CCS to model dynamic behaviors of a WS-BPEL business process. Then, by putting two processes of CCS which represent two WS-BPEL WSOs in parallel, compatibility of the two processes can be checked using theory of CCS. A variant form of Process Algebra called Pi-Calculus is also used as the similar solution to formalizing WSO\cite{FormBPELPI}.

Indeed, there are almost numerous research works on formalization for WSO and even workflow. To enumerate all the related works is unreachable. We give the representable works using different formal tools as follows. These works adopt the similar solution to giving WSO a formal foundation. \cite{FormBPELWP} analyzes BPEL using a framework composed of workflow patterns and communication patterns and makes a comparison of BPEL, WSFL, XLANG, etc. \cite{FormBPEL2} formally defines an abstract executable semantics for BPEL based on an abstract state machine (ASM). \cite{FormBPEL3} transforms all primitive and structured activities of the BPEL into a composable Timed Automata (TA) for verification. \cite{FormWSOLOTOS} transforms all control flow patterns into LOTOS for validation. \cite{FormBPELModel} translates a WSO described by BPEL into an Extended Finite-state Automaton (EFA) and uses the SPIN model checker for the representation of the EFA model for verification. \cite{FormWSOCalculus} defines a calculus called Calculus for Orchestration of Web Services (COWS) to verify a WSO.

About formalization for an interface of a WSO, especially an observable view, there are only a few works. \cite{InterorgaWorkfowPView1} and \cite{InterorgaWorkfowPView2} defines a process view for an cross-organizational workflow and discuss the usages of such a process view. But the model is almost a conceptual one and efforts still need to be done to comfort a WS Composition background. \cite{FormAbstractProcesses} points out that the abstract process in WS-BPEL does not prevent complex computations and is unnecessarily complex. They propose some restrictions on data manipulation constructs in an abstract process and also introduce a logic framework to verify such a restricted abstract process.

\subsubsection{Formalization for WSC}

Since WSC description language, such as WS-CDL, does not have formal verification support. Formalization for WSC is to give WSC a formal semantics. Such an approach is also to establish maps between entities of a WSC description language and those of a formal tools to provide formal semantics support.

\cite{FormWSCPNet1} proposes a Petri Net approach for the design and analysis of WSC Using WS-CDL\cite{WS-CDL} as the language for describing WSC and Petri Net as a formalism to simulate and validate WSC. To capture timed and prioritized interactions, the Petri Net used is a kind of prioritized version of Time Petri Net (PTPN). \cite{FormWSCPNet2} also gives WSCI\cite{WSCI} a formal semantics based on Petri Net.

\cite{FormWSCPA} presents the semantics of WS-CDL in terms of Process Algebra CSP. Therefore, all the properties of a WSC to be checked can be verified with a CSP framework. \cite{ForWSC1} formalizes WSCI based on Process Algebra CCS and can check whether two or more WSs are compatible.

There are also many works to give WSC a formal semantics. \cite{WSCLogic} presents a formal model of WS-CDL based on a Spatio-Temporal Logic which can be used to reason on properties interested. \cite{FormWSCModel} proposes a denotational semantics model for WSCI. \cite{FormWSCValidation} develops a relational calculus to simulate and validate a WSC described by WS-CDL. \cite{WSCandContract} relates theory of contracts and WSC with a notion of choreography conformance.

\subsubsection{Formalization for Relationship between WSO and WSC}

This topic may include the compatibility verification of the interacting WSs or WSOs, the conformance verification of a WSC and its partner WSs or WSOs, automatic generation with correct assurance of a WSC from existing interface definitions of WSs or WSOs, refinement with correct assurance of interface definitions of WSs or WSOs from an existing WSC.

Formalization works of WSO and WSC can be used to verify properties, such as correctness, of either WSO or WSC separately at modeling time or design time. \cite{WSOandWSCCSP} and \cite{WSOandWSCCSP2} use Process Algebra CSP as a formal basis for verifying the behavioral consistency among abstract and executable processes together with choreographic descriptions. \cite{WSOandWSC2}\cite{WSOandWSC3} transform a WSO described by WS-BPEL and a WSC defined by WS-CDL both into two process algebra, and uses the form of bisimulation relation to define a notion of conformance between the WSO and the WSC. There are also some other works to unify WSO and WSC. \cite{WSOandWSCAutomata} uses Reo coordination language and constraint automata to derive a natural correspondence relationship between WSO and WSC. And also exception handling and finalization/compensation are used to connect a WSO and a WSC\cite{WSOandWSCException}. \cite{BPEL2WSCDL} and \cite{WSCDL2BPEL} make that BPEL specifications can be converted into WS-CDL specifications, and vice versa.

This paper involves discussing a system with formal semantics support, which not only can be used for verifications and simulations at design time, including verifying properties of WSO and WSC itself, and that of the relationship between WSO and WSC, but also can serve as a WS Composition runtime system with correctness analysis support indeed. BSPL\cite{BSPL}\cite{BSPL2} is an information-driven interaction-oriented programming language and can be used to model WSC, which is somewhat similar to our approach in this paper, but this is another topic.

\section{An Architecture of WS Composition Runtime}

In this section, we firstly analyze the natures of WSO and WSC. Based on the analysis, design decisions on WS Composition runtime are made. Finally, we design the architecture of WS Composition runtime, within which has a natural relationship among WS, WSO and WSC.

\subsection{WSO and WSC}\label{WSO and WSC}

A WS is a distributed software component with transport protocol--SOAP, interface description by WSDL, and can be registered into UDDI to be searched and discovered by its customers.

A WSO orchestrates WSs existing on the Web into a process through the so-called control flow constructs. That is, within a WSO, there are a collection of atomic function units called activities with control flows to manipulate them. So, the main ingredients of a WSO are following.

\begin{itemize}
  \item Inputs and Outputs: At the start time of a WSO, it accepts some inputs. And it sends out outcomes at the end of its execution.
  \item Information and Variable Definitions: A WSO has local states which maybe transfer among activities. Finally, the local states are sent to WSs outside by activities in the form of messages. In turn, activities receiving message outside can alter the local states.
  \item Activity Definitions: An activity is an atomic unit with several pre-defined function kinds, such as invoking a WS outside, invoking an application inside, receiving a request from a customer inside/outside, local variable assignments, etc.
  \item Control Flow Definitions: Control flow definitions give activities an execution order. In terms of structural model based control flow definitions, control flows are the so-called structural activities which can be sequence activity, choice activity, loop activity, parallel activity and their variants.
  \item Binding WS Information: Added values of WS Composition are the so called recursive composition, that is, a WSO orchestrating existing WSs is published as a new WS itself too. A WSO interacts with other WSs outside through this new WS (that is detailed in \ref{in and out}).
\end{itemize}

In Fig.\ref{Fig.5}, the user agent business process is modeled as UserAgent WSO described by WS-BPEL, which is described in Appendix.\ref{XMLDescription}.

The interface WS for UserAgent WSO is called UserAgent WS described by WSDL, which also can be found in Appendix.\ref{XMLDescription}.

A WSC defines the external interaction behaviors and serves as a contract or a protocol among WSs. The main ingredients of a WSC are as following.

\begin{itemize}
  \item Parter Definitions: They defines the partners within a WSC including the role of a partner acting as and relationships among partners.
  \item Information and Variable Definitions: A WSC may also have local states exchanged among the interacting WSs.
  \item Interactions among Partners: Interaction points and interaction behaviors are defined as the core contents in a WSC.
\end{itemize}

In the buying books example, the WSC between user agent and bookstore (exactly UserAgentWS and BookStoreWS) called BuyingBookWSC is described by WS-CDL, which can be found in Appendix.\ref{XMLDescription}.

The WSO and the WSC define two different aspects of WS Composition. Their relationships as Fig.\ref{Fig.1} illustrates. Note that a WSO may require at least a WSC, but a WSC does not need to depend on a WSO.

\begin{figure}
  \centering
  \includegraphics{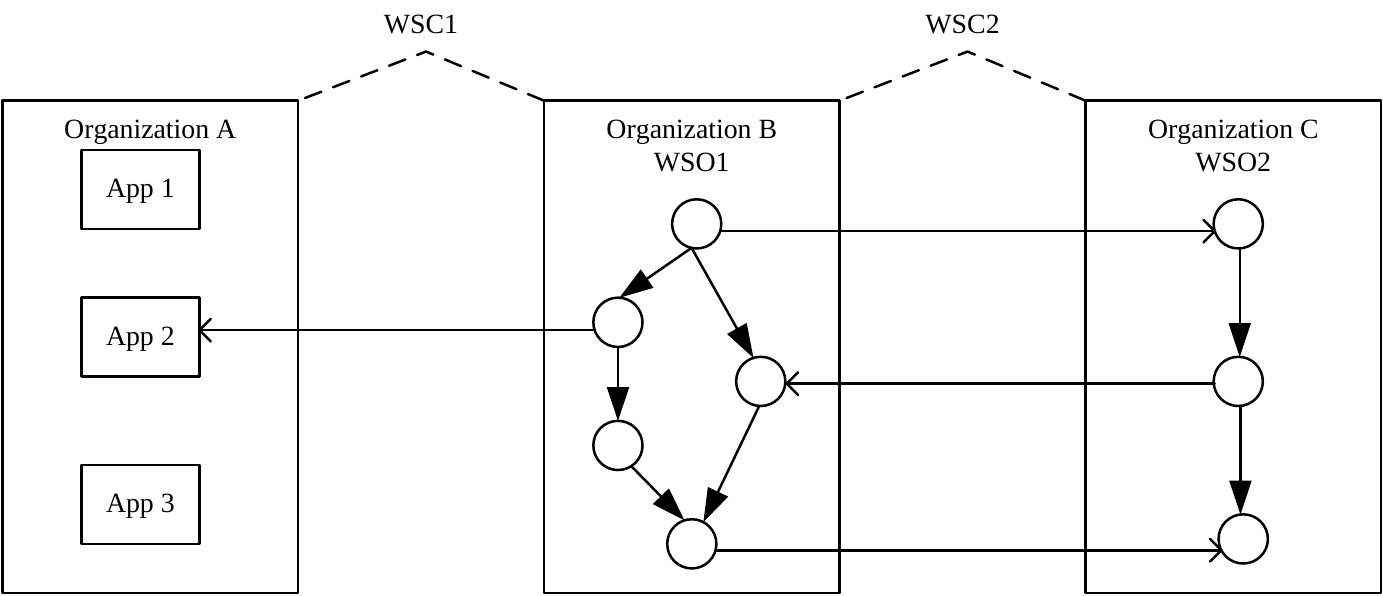}
  \caption{Relationship between WSO and WSC.}
  \label{Fig.1}
\end{figure}

\subsection{Design Decisions on Web Service Composition Runtime}

\subsubsection{Stateless WS or Stateful WS}

In the viewpoint of W3C, a WS itself is an interface or a wrapper of an application inside the boundary of an organization that has a willing to interact with applications outside. That is, a W3C WS has no an independent programming model like other component models and has no needs of containing local states for local computations. Indeed, there are different sounds of developing WS to be a full sense component, Such as OGSI\cite{OGSI}. Incompatibility between W3C WS and OGSI-like WS leads to WSRF\cite{WSRF} as a compromised solution which reserves the W3C WS and develops a notion of WS Resource to model states.

We adopt the ideas of WSRF. That is, let WS be an interface or a wrapper of WSO and let WSO be a special kind WS Resource which has local states and local computations. The interface WS of a WSO reserves ID of the WSO to deliver an incoming message to the WSO and send an outgoing message with the ID attached in order for delivering a call-back message. Further more, a WSO and its WS are one-one binding. When a new incoming message arrives without a WSO ID attached, the WS creates a new WSO and attaches its ID as a parameter. About creation of a WSO by its WS, please refer to \ref{WS Creation}.

\subsubsection{Incoming Messages and Outgoing Messages}\label{in and out}

Just as the name implies, a WS serves as a server to process an incoming message within a C/S framework. But an interaction between a component WS or a WSO requires incoming message and outgoing message pairs. When an interaction occurred, one serves as a client and the other serves as a server. But in the next interaction, the one served as client before may serve as a server and the server becomes a client.

The problem is that, when a WSO (or other kind WS Resource) inside interacts with WSs outside, who is willing to act as the bridge between the WSO inside and WSs outside? When an incoming message arrives, it is easily to be understood that the incoming message is delivered to the WSO by the interface WS. However, how is an outgoing message from a WSO inside to a component WS outside delivered?

In fact, there are two ways to solve the outgoing message. One is the way of WS-BPEL\cite{WS-BPEL}, and the other is that of an early version of WSDL\cite{WSDL}. The former uses a so-called \emph{invoke} atomic activity defined in a WSO to send an outgoing message directly without the assistant of its interface WS. In contrast, the latter specifies that every thing exchenged between resources inside and functions outside must go via the interface WS of the resource inside. Furthermore, in an early edition of WSDL, there are four kind of WS operations are defined, including an \textbf{In} operation, an \textbf{In-Out} operation, an \textbf{Out} operation and an \textbf{Out-In} operation. \textbf{In} operation and \textbf{In-Out} operation receive the incoming messages, while \textbf{Out} operation and \textbf{Out-In} operation deliver the outgoing messages. \textbf{Out} operation and \textbf{Out-In} operation are somewhat strange because a WS is a kind of server in nature. So, in the later versions of WSDL, \textbf{Out} operation and \textbf{Out-In} operation are removed. But the problem of how to process the outgoing message is remaining.

The way of WS-BPEL will cause some confusions in the WS Composition runtime architecture design (see \ref{architecture}). And the way of the early edition of WSDL looks somewhat strange. So, our way of processing outgoing message is a compromise of the above two ones. That is, the outgoing messages from an internal WSO to an external resource, must go via the WS of the internal WSO. But the WS does not need to declare operations for processing the outgoing messages in the WSDL definitions.

\subsubsection{Functions and Enablements of WSC}\label{WS Creation}

A WSC acts as a contract or a protocol between interacting WSs. In a viewpoint of business integration requirements, a WSC serves as a business contract to constrain the rights and obligations of business partners. And from a view of utilized technologies, a WSC can be deemed as a communication protocol which coordinates the interaction behaviors of involved WSs.

About the enablements of a WSC, there are also two differently enable patterns. One is a concentrated architecture and the the other is a distributed one.

The concentrated way considers that the enablements of a WSC must be under supervision of a thirdly authorized party or all involved partners. An absolutely concentrated way maybe require that any operation about interacting WSs must be done by the way of a supervisor. This way maybe cause the supervisor becoming a performance bottleneck when bulk of interactions occur, but it can bring trustworthiness of interaction results if the supervisor is trustworthy itself.

The distributed way argues that each WS interacts among others with constraints of a WSC and there is no need of a supervisor. It is regarded that WSs just behave \emph{correctly} to obey to a WSC and maybe take an example of enablements of open Internet protocols. But there are cheating business behaviors of an intendedly \emph{incorrect} WS, that are unlike almost purely technical motivations of open Internet protocols.

In this paper, we use a hybrid enablements of WSC. That is, when a WSC is contracted (either contracted dynamically at runtime or contracted with human interventions at design time) among WSs and enabled, the WSC creates the partner WSs at the beginning of enablements. And then the WSs interact with each other \emph{freely}\footnote[1]{It means that the WSs interact without the supervision of the WSC.} and \emph{correctly}\footnote[2]{It means that the interaction obeys the constraint of the WSC.}. About the cheating behaviors of a WS are out of the scope of this paper.

\subsection{A WS Composition Runtime Architecture}\label{architecture}

Based on the above introductions and discussions, we design an architecture of WS Composition runtime as Fig.\ref{Fig.2} shows. Fig.\ref{Fig.2} illustrates the typical architecture of a WS Composition runtime. We explain the compositions and their relationships in the following. There are four components: WSO, WS, WSC and applications inside.

\begin{figure}
  \centering
  \includegraphics{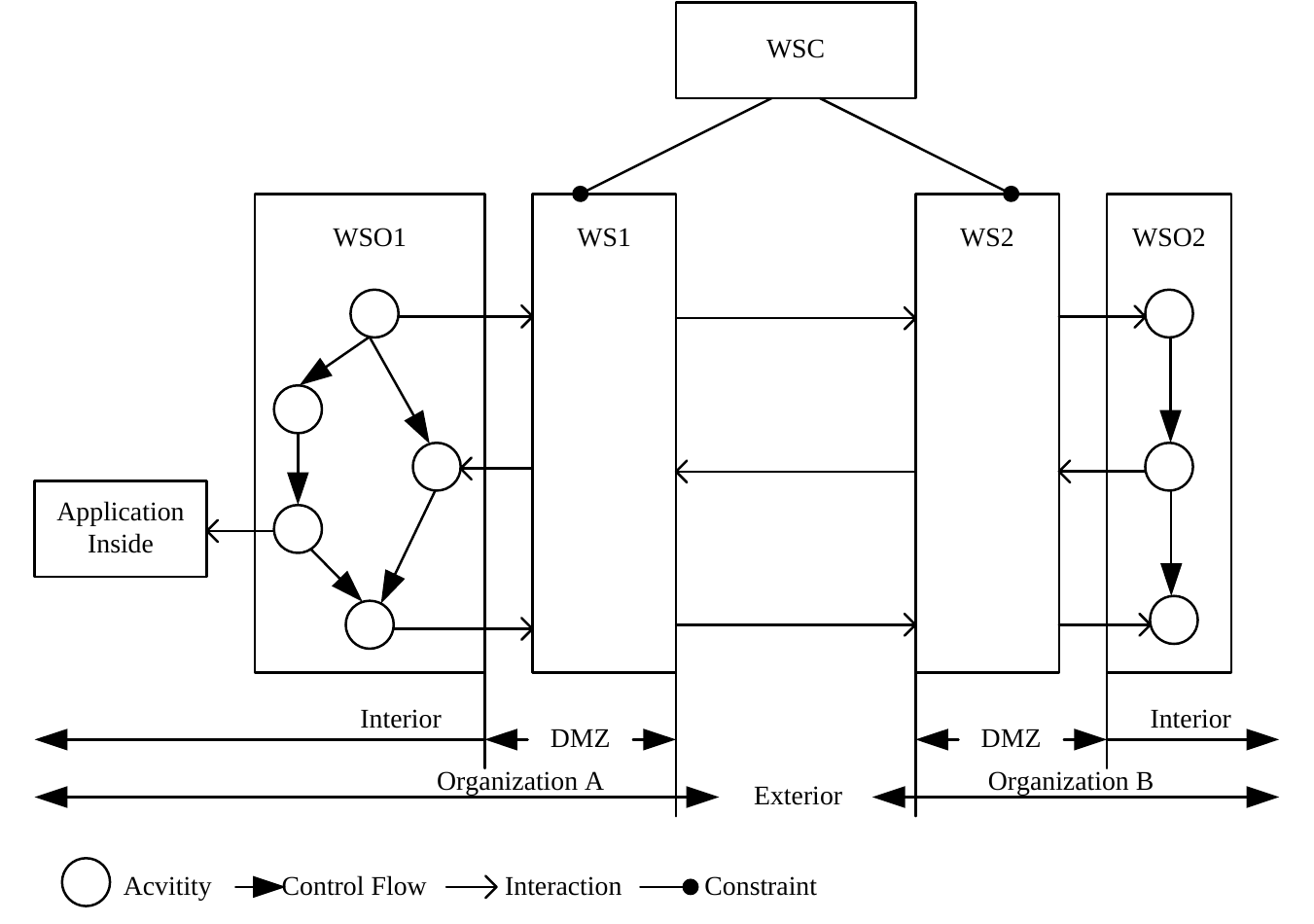}
  \caption{An Architecture of A WS Composition Runtime.}
  \label{Fig.2}
\end{figure}

The functions and ingredients of a WSO are introduced in \ref{WSO and WSC}. Usually it has a collection of activities that may interact with partner WSs outside or applications inside. Enablements of a WSO require a runtime environment which is not illustrated in Fig.\ref{Fig.2}. For examples, execution of a WSO described by WS-BPEL needs a WS-BPEL interpreter (also called WSO engine like ActiveBPEL\cite{ActiveBPEL}) and a WSO modeled in our AB-WSCL also needs an implementation (see \ref{language implementations}). A WSO locates in the interior of an organization. It interacts with applications inside with private exchanging mechanisms and with other partner WSOs outside via its interface WS.

Applications inside may be any legacy application or any newly developed application within the interior of a organization. These applications can be implemented in any technical framework and provide interfaces to interact with other applications inside, including a WSO. Interactions between a WSO and a application inside may base on any private communication mechanism, such as local object method call, RPC, RMI, etc, which depends on technical framework adopted by the application.

An interface WS acts as an interface of a WSO to interact with partner WSs outside. A WSO is with an one-to-one binding to its interface WS and is created by its interface WS at the time of first interaction with exterior. (Relationship of a WSO and its interface WS please refer to \ref{in and out}). Enablements of a WS also require a runtime support usually called SOAP engine like AXIS SOAP Engine\cite{AXIS} which implies a HTTP server installed to couple with HTTP requests. A WS and its runtime support locate at demilitarized zone (DMZ) of an organization which has different management policies and different security policies to the interior of an organization.

A WSC acts as a contract or a protocol of partner WSs (see \ref{WSO and WSC}). When a WSC is enabled, it creates all partner WSs at their accurate positions (see \ref{WS Creation}). Enablements of a WSC also require a runtime support to interpret the WSC description language like WS-CDL. A WSC and its support environment can be located at a thirdly authorized party or other places negotiated by the partners.

\section{A Language of WS Composition--Actor-Based Web Service Composition Language, AB-WSCL}

Actor\cite{Actor1}\cite{Actors}\cite{ConcurrentObject} is a basic computation model that captures the natures of concurrency in distributed computing. There are many efforts to abstract at a high-level from aspects of distributed computing, such as policy management\cite{ActorCustandCompo}, interaction policies\cite{ActorInteraction}, resource management\cite{ActorResource}, communication and coordination of agents\cite{ActorAgent}, worldwide computing\cite{ActorWWW}, etc. Our works are somewhat based on these high-level abstractions and make customizations for special requirements of WS Composition.

In this section, we introduce an actor-based language called AB-WSCL\cite{ABWSCL} to support the WS Composition runtime architecture in \ref{architecture}. Firstly, basics of an actor are introduced. Then we give the definitions of ingredients of AB-WSCL, including activity actor (AA), web service orchestration (WSO), web service (WS), and web service choreography (WSC). An example of WS Composition programmed with AB-WSCL is illustrated and some issues about AB-WSCL are discussed. Finally, we discuss about implementation of AB-WSCL based on an actor runtime called Actor Foundry\cite{ActorFoundry}.

AB-WSCL adopts a Java-like syntax and the syntax of Java is not explained any more.

\subsection{Actors}

An actor acts as an atomic function unit of concurrent computation and is able to model elements of the architecture in \ref{architecture}. We use an actor as a distributed concurrent object, as illustrated in \cite{ConcurrentObject}.

An actor is a concurrent object that encapsulates a set of states, a control thread and a set of local computations. It has a unique mail address and maintains a mail box to accept messages sent by other actors. Actors do local computations by means of processing the messages stored in the mail box sequentially and block when their mail boxes are empty.

During processing a message in mail box, an actor may perform three candidate actions: (1)(\textbf{send} action)sending messages asynchronously to other actors by their mail box addresses; (2)(\textbf{create} action) creating new actors with new behaviors; (3)(\textbf{ready} action) becoming ready to process the next message from the mail box or block if the mail box is empty. The illustration of an actor model as shows in Fig.\ref{Fig.3} which is first shown in \cite{Actors}.

\begin{figure}
  \centering
  \includegraphics{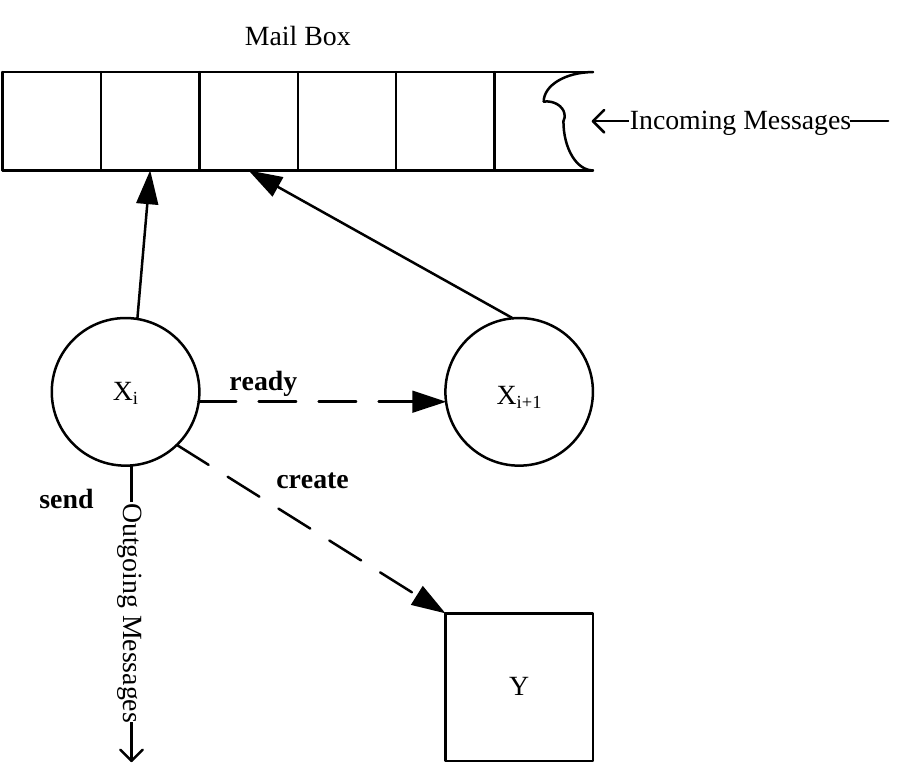}
  \caption{Model of an Actor.}
  \label{Fig.3}
\end{figure}

Note that messages are sent asynchronously among actors by default, but synchronization\cite{ActorSyn} can also be achieved among actors.

Following the action model of an actor in \cite{ActorCustandCompo}, in this paper, we also use a signal-notification pair to factor actions of an actor, as illustrated in Fig.\ref{Fig.4}.

\begin{figure}
  \centering
  \includegraphics{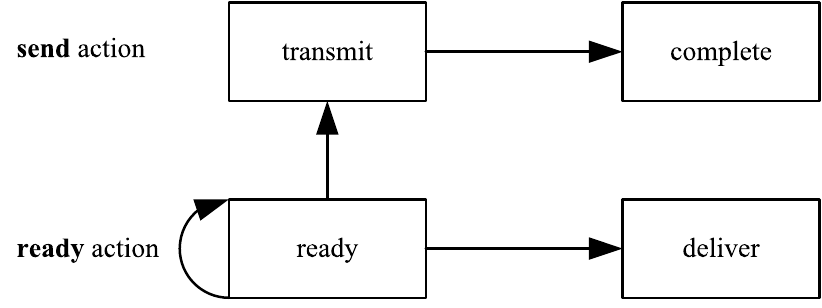}
  \caption{Model of Actor Actions.}
  \label{Fig.4}
\end{figure}

An actor sends a \textbf{transmit} signal and blocks until a \textbf{continue} notification is received. Similarly, an actor sends a \textbf{ready} signal and blocks then may be resumed by a \textbf{deliver} notification, or sending a \textbf{transmit} signal. Differently, an actor creates other actors without any event generated. This model is used to capture the interactions among actors with synchronization support in the following sections.

\subsection{Activity Actor, AA}

An activity is an atomic function unit of a WSO and is managed by the WSO. We use an actor called activity actor (AA) to model an activity.

An AA has a unique name, local information and variables to contain its states, and local computation procedures to manipulate the information and variables. An AA is always managed by a WSO and it receives messages from its WSO, sends messages to other AAs or WSs via its WSO, and is created by its WSO. Note that an AA can not create new AAs, it can only be created by a WSO. That is, an AA is an actor with a constraint that is without \textbf{create} action.

We give the abstract syntax of an AA as follows.

---------------------------------------------------------------------------------------------------

AA ::= \emph{AA} \quad name \{

  \quad\quad \emph{WSO} \quad wso-ref

  \quad\quad local-variable-declaration

  \quad\quad [init(args)\{AA-action*\}]

  \quad\quad AA-method*

\}

AA-method ::= [local] AA-method-name (args) if condition \{

  \quad\quad AA-action*

\}

AA-action ::= local-variable-assignment

\quad\quad\quad\quad\quad\quad$|$ wso-ref $\leftarrow$ WSO-method-name (args)

\quad\quad\quad\quad\quad\quad$|$ interacting with applications inside

---------------------------------------------------------------------------------------------------

Where the symbol * means that there are several AA objects with in a WSO and it means the same as the following. The first AA is the AA object and the second \emph{AA} is the key word, the expression \emph{actor-name $\leftarrow$ actor-method-name (args)} denotes sending a message \emph{args} to the actor with a name \emph{actor-name} by way of invoking the method with a name \emph{actor-method-name}.

As the above syntax described, an AA has a unique \emph{name} and a reference of its WSO (\emph{wso-ref}). Its local variables can be declared by a type system, such as XML Schema\cite{XMLSchema} or Java, and can be an actor type. It has an optional initializing method which is executed when the actor is created. It also has a set of methods which can be invoked by other actors. Each method can do some actions, including local variable assignment, interacting with applications inside, and sending to its WSO a message, through an invocation of a method of its WSO.

Note that an AA does not have \emph{create} action, Fig.\ref{Fig.4} also gives the signal-notification pairs of actions of an AA.


\subsection{Web Service Orchestration, WSO}

A WSO includes a set of AAs and acts as the manager of the AAs. The management operations may be creating a member AA, acting as a bridge between AAs and acting as a bridge between AAs and WSs outside.

We give the abstract syntax of a WSO as follows.

---------------------------------------------------------------------------------------------------

WSO ::= \emph{WSO} \quad name \{

  \quad\quad \emph{AA} \quad aa-ref*

  \quad\quad \emph{WS} \quad ws-ref

  \quad\quad local-variable-declaration

  \quad\quad [init(args)\{WSO-action*\}]

  \quad\quad WSO-method*

\}

WSO-method ::= [local] WSO-method-name (args) if condition \{

  \quad\quad WSO-action*

\}

WSO-action ::= local-variable-assignment

\quad\quad\quad\quad\quad\quad$|$ aa-ref $\leftarrow$ AA-method-name (args)

\quad\quad\quad\quad\quad\quad$|$ ws-ref $\leftarrow$ WS-method-name (args)

\quad\quad\quad\quad\quad\quad$|$ aa-ref := new \emph{AA} (args)

---------------------------------------------------------------------------------------------------

The first WSO is the WSO object and the second \emph{WSO} is the key word of the WSO object. Similar to an AA, a WSO also has local variables, an \emph{init} method, and a set of methods. Differently, A WSO has a set of AAs, and also an interface WS actor. The actions of a WSO can be local variable assignment, sending messages to AAs, sending a message to its interface WS, and also creating a new AA.

The signal-notification pairs of actions of a WSO are also as Fig.\ref{Fig.4} shows.

\subsection{Web Service, WS}

A WS is an actor that has the characteristics of an ordinary actor. It acts as a communication bridge between the inner WSO and the external partner WS and creates a new WSO when it receives a new incoming message.

The abstract syntax of a WS is following.

---------------------------------------------------------------------------------------------------

WS ::= \emph{WS} \quad name \{

  \quad\quad \emph{WSO} \quad wso-ref

  \quad\quad \emph{WS} \quad ws-ref

  \quad\quad local-variable-declaration

  \quad\quad [init(args)\{WS-action*\}]

  \quad\quad setPartner (args) if condition \{

  \quad\quad \quad\quad WS-action*

  \quad\quad \}

  \quad\quad WS-method*

\}

WS-method ::= [local] WS-method-name (args) if condition \{

  \quad\quad WS-action*

\}

WS-action ::= local-variable-assignment

\quad\quad\quad\quad\quad\quad$|$ wso-ref $\leftarrow$ WSO-method-name (args)

\quad\quad\quad\quad\quad\quad$|$ ws-ref $\leftarrow$ WS-method-name (args)

\quad\quad\quad\quad\quad\quad$|$ wso-ref := new \emph{WSO} (args)

---------------------------------------------------------------------------------------------------

The first WS is the WSO object and the second \emph{WS} is the key word of the WS object. A WS has not only local variables, an \emph{init} method, and a set of methods, but also a required \emph{setPartner} method to be invoked by a WSC to get a reference of its partner WS. Its actions include local variable assignment, sending a message to its WSO, sending a message to its partner WS and creating a new WSO.

The signal-notification pairs of actions of a WS are also as Fig.\ref{Fig.4} shows.

\subsection{Web Service Choreography, WSC}

A WSC actor creates partner WSs as some kinds roles and set each WS to the other one as their partner WSs.

The abstract syntax of a WSC is following.

---------------------------------------------------------------------------------------------------

WSC ::= \emph{WSC} \quad name [role partner-role-name*]\{

  \quad\quad \emph{WS} \quad ws-ref-1

  \quad\quad \emph{WS} \quad ws-ref-2

  \quad\quad local-variable-declaration

  \quad\quad [init(args)\{WS-action*\}]

  \quad\quad WSC-method*

\}

WSC-method ::= [local] WSC-method-name (args) if condition \{

  \quad\quad WSC-action*

\}

WSC-action ::= local-variable-assignment

\quad\quad\quad\quad\quad\quad$|$ ws-ref-1 $\leftarrow$ WS-method-name (args)

\quad\quad\quad\quad\quad\quad$|$ ws-ref-2 $\leftarrow$ WS-method-name (args)

\quad\quad\quad\quad\quad\quad$|$ ws-ref-1 := new \emph{WS} (args) as partner-role-name

\quad\quad\quad\quad\quad\quad$|$ ws-ref-2 := new \emph{WS} (args) as partner-role-name

\quad\quad\quad\quad\quad\quad$|$ ws-ref-1 $\leftarrow$ setPartner (args)

\quad\quad\quad\quad\quad\quad$|$ ws-ref-2 $\leftarrow$ setPartner (args)

---------------------------------------------------------------------------------------------------

Similarly, the first WSC is the WSC object and the second \emph{WSC} is also the key word the WSC object. With a customization to an ordinary actor, a WSC contains two interacting WSs. The actions of a WSC includes creating each WS, setting the partner to each WS, and sending other messages to each WS.

The signal-notification pairs of actions of a WSC are also as Fig.\ref{Fig.4} shows.

\subsection{Implementation for Buying Books Example by AB-WSCL}

Using the architecture in Fig.\ref{Fig.2}, we get an implementation of the buying books example as showing in Fig.\ref{Fig.6}. In this implementation, there are one WSC (named BuyingBookWSC), two WSs (one is named UserAgentWS, the other is named BookStoreWS), two WSOs (one is named UserAgentWSO, the other is named BookStoreWSO), and two set of AAs. The set of AAs belong to UserAgentWSO including RequstLBAA, ReceiveLBAA, SendSBAA, ReceivePBAA and PayBAA, and the other set of AAs belong to BookStoreWSO including ReceiveRBAA, SendLBAA, ReceiveSBAA, SendPBAA, and GetP\&ShipBAA.

\begin{figure}
  \centering
  \includegraphics{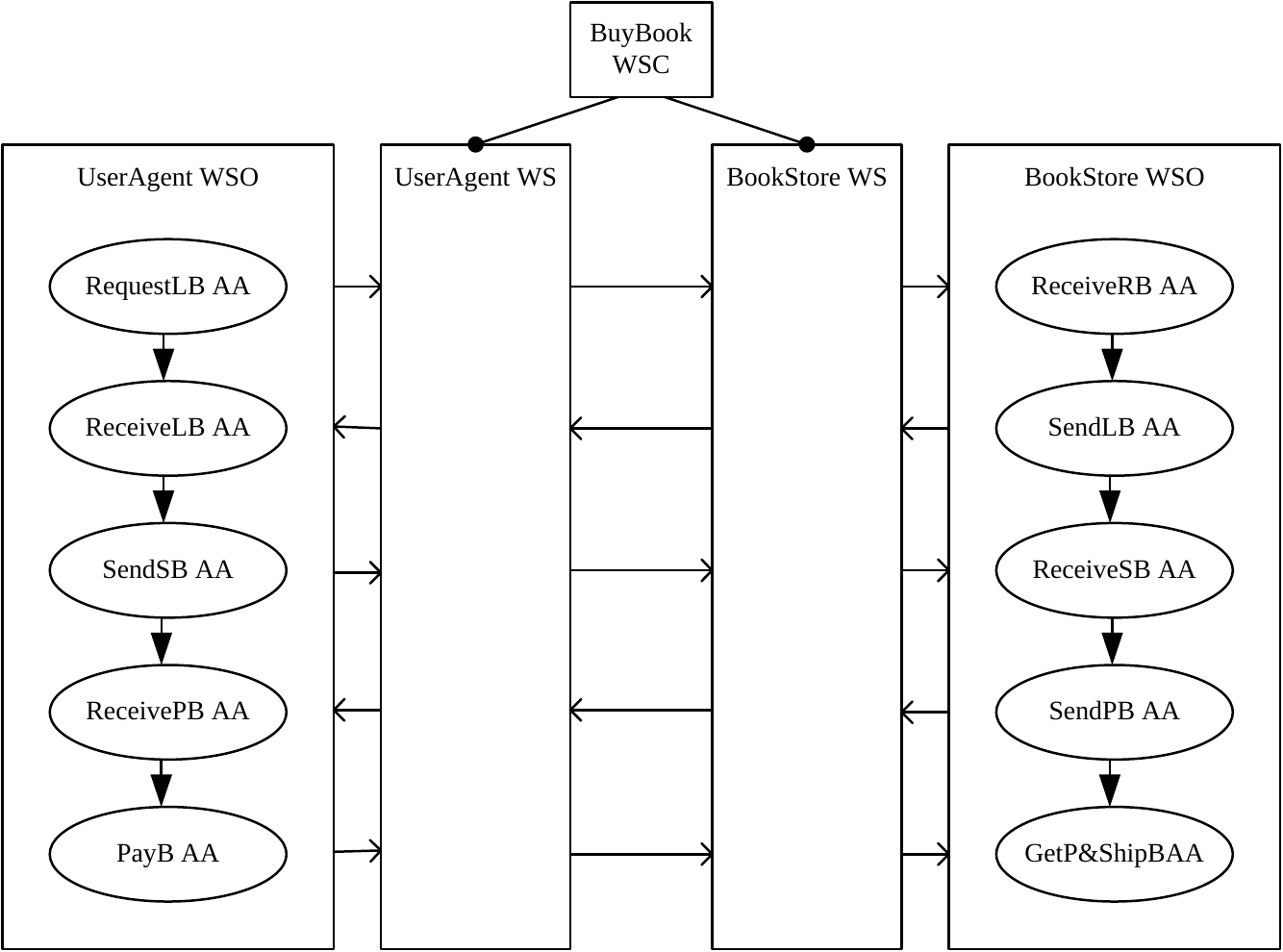}
  \caption{Implementation of the Example.}
  \label{Fig.6}
\end{figure}

The detailed implementations of entities in Fig.\ref{Fig.6} is following. We give the implementations of user agent side, and those of book store side can be gotten in contrary by the reader self. And an implementation of BuyingBookWSC is also shown.

The implementations of user agent side includes AAs (RequstLBAA, ReceiveLBAA, SendSBAA, ReceivePBAA and PayBAA), UserAgentWSO, UserAgentWS.

Implementation of RequstLBAA:

---------------------------------------------------------------------------------------------------

AA RequstLBAA\{

\quad\quad \emph{WSO} wso-ref

\quad\quad init(\emph{WSO} wso)\{

\quad\quad\quad\quad wso-ref := wso

\quad\quad\quad\quad other-local-computations

\quad\quad\}

\quad\quad requestLBFromCustomer()if true \{

\quad\quad\quad\quad other-local-computations

\quad\quad\quad\quad wso-ref $\leftarrow$ requestLB()

\quad\quad \}

\}

---------------------------------------------------------------------------------------------------

Implementation of ReceiveLBAA:

---------------------------------------------------------------------------------------------------

AA ReceiveLBAA\{

\quad\quad \emph{WSO} wso-ref

\quad\quad List books

\quad\quad init(\emph{WSO} wso)\{

\quad\quad\quad\quad wso-ref := wso

\quad\quad\quad\quad other-local-computations

\quad\quad\}

\quad\quad\ receiveLB(List bs) if true\{

\quad\quad\quad\quad books := bs

\quad\quad\quad\quad other-local-computations

\quad\quad\}

\}

---------------------------------------------------------------------------------------------------

Implementation of SendSBAA:

---------------------------------------------------------------------------------------------------

AA SendSBAA\{

\quad\quad \emph{WSO} wso-ref

\quad\quad List selectedBooks

\quad\quad init(\emph{WSO} wso)\{

\quad\quad\quad\quad wso-ref := wso

\quad\quad\quad\quad other-local-computations

\quad\quad\}

\quad\quad\ receiveSBFromCustomer(List sb) if true\{

\quad\quad\quad\quad selectedBooks := sb

\quad\quad\quad\quad other-local-computations

\quad\quad\quad\quad wso-ref $\leftarrow$ sendSB (selectedBooks)

\quad\quad\}

\}

---------------------------------------------------------------------------------------------------

Implementation of ReceivePBAA

---------------------------------------------------------------------------------------------------

AA ReceivePBAA\{

\quad\quad \emph{WSO} wso-ref

\quad\quad float prices

\quad\quad init(\emph{WSO} wso)\{

\quad\quad\quad\quad wso-ref := wso

\quad\quad\quad\quad other-local-computations

\quad\quad\}

\quad\quad\ receivePB(float p) if true\{

\quad\quad\quad\quad prices := p

\quad\quad\quad\quad other-local-computations

\quad\quad\}

\}

---------------------------------------------------------------------------------------------------

Implementation of PayBAA:

---------------------------------------------------------------------------------------------------

AA PayBAA\{

\quad\quad \emph{WSO} wso-ref

\quad\quad init(\emph{WSO} wso)\{

\quad\quad\quad\quad wso-ref := wso

\quad\quad\quad\quad other-local-computations

\quad\quad\}

\quad\quad payBFromCustomer() if true\{

\quad\quad\quad\quad other-local-computations

\quad\quad\quad\quad wso-ref $\leftarrow$ PayB()

\quad\quad \}

\}

---------------------------------------------------------------------------------------------------

Implementation of UserAgentWSO:

---------------------------------------------------------------------------------------------------

WSO UserAgentWSO\{

\quad\quad \emph{AA} requestLBAA

\quad\quad \emph{AA} receiveLBAA

\quad\quad \emph{AA} sendSBAA

\quad\quad \emph{AA} receivePBAA

\quad\quad \emph{AA} payBAA

\quad\quad \emph{WS} ws-ref

\quad\quad List books

\quad\quad List selectedBooks

\quad\quad float prices

\quad\quad init(\emph{WS} ws)\{

\quad\quad\quad\quad ws-ref := ws

\quad\quad\quad\quad requestLBAA := new RequestLBAA(self)

\quad\quad\quad\quad receiveLBAA := new ReceiveLBAA(self)

\quad\quad\quad\quad sendSBAA := new SendSBAA(self)

\quad\quad\quad\quad receivePBAA := new ReceivePBAA(self)

\quad\quad\quad\quad payBAA := new PayBAA(self)

\quad\quad\quad\quad other-local-computations

\quad\quad\}

\quad\quad requestLB() if true\{

\quad\quad\quad\quad other-local-computations

\quad\quad\quad\quad ws-ref $\leftarrow$ requestLB()

\quad\quad \}

\quad\quad receiveLB(List bs) if true\{

\quad\quad\quad\quad books := bs

\quad\quad\quad\quad other-local-computations

\quad\quad\quad\quad receiveLBAA $\leftarrow$ receiveLB(books)

\quad\quad \}

\quad\quad sendSB(sb) if true\{

\quad\quad\quad\quad selectedBooks := sb

\quad\quad\quad\quad other-local-computations

\quad\quad\quad\quad ws-ref $\leftarrow$ sendSB(selectedBooks)

\quad\quad \}

\quad\quad receivePB(float pb) if true\{

\quad\quad\quad\quad prices := pb

\quad\quad\quad\quad other-local-computations

\quad\quad\quad\quad receivePBAA $\leftarrow$ receivePB(prices)

\quad\quad \}

\quad\quad payB() if true\{

\quad\quad\quad\quad other-local-computations

\quad\quad\quad\quad ws-ref $\leftarrow$ payB()

\quad\quad \}

\}

---------------------------------------------------------------------------------------------------

Implementation of UserAgentWS:

---------------------------------------------------------------------------------------------------

WS UserAgentWS\{

\quad\quad \emph{WSO} wso-ref

\quad\quad \emph{WS} ws-ref

\quad\quad List books

\quad\quad List selectedBooks

\quad\quad float prices

\quad\quad init()\{

\quad\quad\quad\quad wso-ref := new UserAgentWSO(self)

\quad\quad\quad\quad other-local-computations

\quad\quad\}

\quad\quad setPartner(\emph{WS} ws) if true\{

\quad\quad\quad\quad ws-ref = ws

\quad\quad\quad\quad other-local-computations

\quad\quad\}

\quad\quad requestLB() if true\{

\quad\quad\quad\quad other-local-computations

\quad\quad\quad\quad ws-ref $\leftarrow$ requestLB()

\quad\quad \}

\quad\quad receiveLB(List bs) if true\{

\quad\quad\quad\quad books :=bs

\quad\quad\quad\quad other-local-computations

\quad\quad\quad\quad wso-ref $\leftarrow$ receiveLB(books)

\quad\quad \}

\quad\quad sendSB(sb) if true\{

\quad\quad\quad\quad selectedBooks := sb

\quad\quad\quad\quad other-local-computations

\quad\quad\quad\quad ws-ref $\leftarrow$ sendSB(selectedBooks)

\quad\quad \}

\quad\quad receivePB(float pb) if true\{

\quad\quad\quad\quad prices := pb

\quad\quad\quad\quad other-local-computations

\quad\quad\quad\quad wso-ref $\leftarrow$ receivePB(prices)

\quad\quad \}

\quad\quad payB() if true\{

\quad\quad\quad\quad other-local-computations

\quad\quad\quad\quad ws-ref $\leftarrow$ payB()

\quad\quad \}

\}

---------------------------------------------------------------------------------------------------

Implementation of BuyingBookWSC:

---------------------------------------------------------------------------------------------------

WSC BuyingBookWSC role user, role seller\{

\quad\quad \emph{WS} ws-ref-1

\quad\quad \emph{WS} ws-ref-2

\quad\quad init()\{

\quad\quad\quad\quad ws-ref-1 := new UserAgentWS() as user

\quad\quad\quad\quad ws-ref-2 := new BookStoreWS() as seller

\quad\quad\quad\quad ws-ref-1 $\leftarrow$ setPartner(ws-ref-2)

\quad\quad\quad\quad ws-ref-2 $\leftarrow$ setPartner(ws-ref-1)

\quad\quad\quad\quad other-local-computations

\quad\quad\}

\quad\quad\ other-WSC-methods

\}

---------------------------------------------------------------------------------------------------

\subsection{Issues of AB-WSCL}

We discuss some issues of AB-WSCL introduced in this section, including providing some kind of pre-defined AA, implementation of control flow in a WSO and lifetime of components of AB-WSCL.

\subsubsection{AA Types}

In AB-WSCL, we model an activity as an AA, which is an actor with a constraint that can not create other AAs. AAs are always under the control of a WSO, especially can only be created by its WSO. But an AA has local states, procedures for local computations and can receive messages from other actors and can also send messages to others. So, an AA can do anything except for \emph{create} action and should be able to be used to implement any kind of functions.

From a low-level technical viewpoint, AAs can have a number of fixed functions, such as sending a message to a WS outside, receiving a message from a WS outside, interacting with an application inside or a human, local variable assignments and other inner processing within a WSO. In fact, we can use the definition of AA in AB-WSCL to further pre-define some kind of AAs for convenience. Let us leave this work for future (see \ref{Future Works}).

From a view of high-level business requirements, AB-WSCL can be used to model business process integrations. In some specific business domains, such as book selling, traveling, food manufacturing, etc, it is possible to abstract some kind of business functions to be modeled as some pre-defined AA types by the assistant of domain experts. We also leave this work for future (see \ref{Future Works}).

\subsubsection{Control Flows within a WSO}

Orchestration means that all AAs are orchestrated by control flows among them. That is, the executing of AAs should be in some orders, such as in sequence, in parallel, in choice, or in loop, etc. Sequence, parallel, choice and loop are four basic constructs for an orchestration. We show the solutions of AB-WSCL for the above four constructs one by one.

Firstly, we discuss about sequence construct. AB-WSCL is a high-level concurrent computing language while an actor is a basic concurrent computation model. Actors are basic components in AB-WSCL, that is, AA, WSO, WS, WSC are all actors. As Agha\cite{Actors} pointed out, actor computing is intrinsic concurrent and sequential computing can only be implemented by causalities of messages sending among actors.

Parallel is intrinsical to AAs just because AAs are all concurrent actors.

Readers may notice that declarations of methods of actors (AA, WSO, WS, WSC) to receive messages from other actors have the form of \emph{method-name (args) if condition}. There is a \emph{condition} consisted of local states of actors to control whether or not to receive a message for an actor. Using this \emph{condition} can implement the choice construct.

Similarly, loop construct can also be implemented through the \emph{condition} of a method.

\subsubsection{Lifetime of AA, WSO, WS and WSC}

Actors in AB-WSCL, such as AA, WSO, WS, WSC, should have a lifetime under the requirements of business integration. We discuss the creation operation and termination operation of lifetime control.

As introduced above, an actor is created explicitly by a \emph{create} action of another actor. That is, an AA is created by a WSO explicitly, a WSO is created by a WS explicitly, a WS is created by a WSC explicitly, and a WSC may be created explicitly by an application upper layer.

After being created, an actor as an active distributed object runs independently like Fig.\ref{Fig.3} shows. At some exact time, such as neither receiving messages nor sending messages, an actor can be collected by a garbage collection mechanism\cite{GarbageCollection}. This mechanism can be invoked by the user explicitly or can be implemented implicitly by the actor runtime system.

\subsubsection{State Control of AA, WSO, WS and WSC}

An actor has its inner states and some block relations as shown in Fig.\ref{Fig.4}. These states of an actor are modeled from a viewpoint of a universal computation model. That is, there are no states under the requirements of business process integration. For example, an actor can respond to a human being to be suspended or be resumed again. We can customize the common actor model to extend more states and control mechanisms. Let this work be an extension of AB-WSCL to be done in future (see \ref{Future Works}).

\subsubsection{About Multi-party Choreographies}

The WSC defines interaction behaviors of involved WSs. Multi-party choreographies can be captured by a whole WSC or by a series of binary WSCs. Since the WSC actor in AB-WSCL is a binary WSC, multi-party choreoraphies can be programmed by a set of WSC actors.

For simplicity of the following formal semantics, without loss of generality, the binary WSC is adopted. In fact, the WSC actor in AB-WSCL can be easily extended to contain a set of partner WSs to directly support multi-party choreographies. The extension involves the definition of the WS actor and the WSC actor. We just give the abstract syntax of the WS actor and the WSC actor to directly support multi-party choreographies and do not mention this issue any more in the rest of this paper.

The extension of the WS actor is following.

---------------------------------------------------------------------------------------------------

WS ::= \emph{WS} \quad name \{

  \quad\quad \emph{WSO} \quad wso-ref

  \quad\quad \emph{WS} \quad ws-ref* //a set of partner WSs

  \quad\quad local-variable-declaration

  \quad\quad [init(args)\{WS-action*\}]

  \quad\quad setPartner (args) if condition \{

  \quad\quad \quad\quad WS-action*

  \quad\quad \}

  \quad\quad WS-method*

\}

WS-method ::= [local] WS-method-name (args) if condition \{

  \quad\quad WS-action*

\}

WS-action ::= local-variable-assignment

\quad\quad\quad\quad\quad\quad$|$ wso-ref $\leftarrow$ WSO-method-name (args)

\quad\quad\quad\quad\quad\quad$|$ ws-ref $\leftarrow$ WS-method-name (args)

\quad\quad\quad\quad\quad\quad$|$ wso-ref := new \emph{WSO} (args)

---------------------------------------------------------------------------------------------------

The extension of the WSC actor is following.

---------------------------------------------------------------------------------------------------

WSC ::= \emph{WSC} \quad name [role partner-role-name*]\{

  \quad\quad \emph{WS} \quad ws-ref* //a set of involving WSs

  \quad\quad local-variable-declaration

  \quad\quad [init(args)\{WS-action*\}]

  \quad\quad WSC-method*

\}

WSC-method ::= [local] WSC-method-name (args) if condition \{

  \quad\quad WSC-action*

\}

WSC-action ::= local-variable-assignment

\quad\quad\quad\quad\quad\quad$|$ ws-ref $\leftarrow$ WS-method-name (args)

\quad\quad\quad\quad\quad\quad$|$ ws-ref := new \emph{WS} (args) as partner-role-name

\quad\quad\quad\quad\quad\quad$|$ ws-ref $\leftarrow$ setPartner (args)

---------------------------------------------------------------------------------------------------

\subsection{About Implementation of AB-WSCL}\label{language implementations}

Actors in AB-WSCL are all customizations of common actors and can be implemented based on a common actor runtime called Actor Foundry\cite{ActorFoundry}. Actor Foundry is developed by Open Systems Laboratory (OSL) of University of Illinois at Urbana-Champaign as an open source software.

Actor Foundry is not only a common actor runtime, but also an actor development framework that allows further developments and customizations. As follows, we discuss the implementations of actors in AB-WSCL.

Because an AA is a common actor without \emph{create} action, and it can be implemented as an actor class inheriting from the abstract class \emph{osl.manager.ActorImpl} without \emph{create} method (That is, let creating actors method be void).

A WSO and a WSC are all basic actors with specific actions and use \emph{osl.manager.basic.BasicActorImpl} class as their implementation bases.

A WS is also a basic actor, but it adopts SOAP as its transport protocol. This requires a SOAP transport package like osl.transport.soap must be implemented firstly. Then a WS can be implemented by inheriting \emph{osl.manager.basic.BasicActorImpl} class with a transport binding to the SOAP package.

\section{Formal Semantics for AB-WSCL}

In the above sections, we design an architecture of WS Composition runtime and introduce an actor-based language called AB-WSCL to support such an architecture. More importantly, in this section we give the formal semantics of AB-WSCL based on concurrent rewriting semantics for actors\cite{ActorRewriting}\cite{ActorRewriting1}\cite{ActorRewriting2}\cite{ActorCustandCompo}, including semantics for AA, WSO, WS and WSC, especially interaction semantics between WSO and WS, WS and WS, WSO and WSO.

In this section, we use symbol conventions that follow \cite{ActorRewriting}\cite{ActorRewriting1}\cite{ActorRewriting2}\cite{ActorCustandCompo}. The symbols are used below.

$f: D \rightarrow R$ represents a total function $f$ with domain $D$ and ranges over $R$.

$f: D \xrightarrow{\circ} R$ represents a partial function $f$ with domain $D$ and ranges over $R$.

$P_n[S]$ represents the set of all subsets $s$ of $S$ where $|s|\leq n$.

$P_\omega [S]$ represents the set of all finite subsets $s$ of $S$.

A signature of rewriting theory denotes as $(\Sigma, E)$, where $\Sigma$ represents the function symbols of the theory, and $E$ denotes the set of $\Sigma$-equations.

A $\Sigma$-algebra is the set of $V$ where $f: V^n \rightarrow V$.

$T_\Sigma$ denotes the $\Sigma$-algebra of ground $\Sigma$-terms.

$T_{\Sigma, E}$ denotes $\Sigma$-algebra of equivalence classes of ground $\Sigma$-terms modulo $E$.

$T_\Sigma(X)$ denotes the $\Sigma$-algebra of ground $\Sigma$-terms with variables in $X$.

$T_{\Sigma, E}(X)$ denotes $\Sigma$-algebra of equivalence classes of ground $\Sigma$-terms with variables in $X$ modulo $E$.

$t(u_1/x_1, ..., u_n/x_n)$ denotes a substitution of each $u_i$ for each $x_i$ with a term $t\in T_\Sigma(\{ x_1,...,x_n \})$ and terms $u_1,...,u_n$.

$[t]$ denotes the $E$-equivalence class of of $t$.

$(L, R)$ denotes the rewriting rules of a signature $(\Sigma, E)$, where $L$ is the set of labels, and $R$ is the set of 3-tuples with $R\subset L\times T_{\Sigma, E}(X) \times T_{\Sigma, E}(X)$.

$\mathbb{V}$ denotes the sort of values.

$\mathbb{A}$ denotes the sort of actor addresses with $\mathbb{A}\subset \mathbb{V}$.

$\mathbb{S}$ denotes the sort of actor states.

$\mathbb{B}$ denotes the sort of actor states with $\mathbb{B}\subset \mathbb{V}$.

$\mathbb{E}$ denotes the finite set of events.

$\mathbb{P}$ denotes the sort of actor processing states with $\mathbb{P}=\{?,!\}$, where $?$ denotes a \emph{ready} state and $!$ denotes a \emph{running} state.

$\mathbb{F}$ denotes the sort of actor fragments.

$\mathbb{C}$ denotes the sort of configurations.

$p\in\mathbb{P}$ denotes a processing state of an actor.

$a\in\mathbb{A}$ denotes an address of an actor.

$b\in\mathbb{B}$ denotes a behavior of an actor.

$s\in\mathbb{S}$ denotes a state of an actor.

$l\in\mathbb{E}$ denotes the last event generated by an actor.

$t\in\mathbb{A}$ denotes the transition map for an actor.

$\langle (s_1)_{a_1},...,(s_n)_{a_n}, a_i \lhd v_i, ..., a_j \lhd v_j \rangle$ is a typical term in operational semantics of actors, where $(s_i)_{a_i}$ denotes an actor with a state $s_i$ and an address $a_i$, and $a_i \lhd v_i$ denotes a message with destination $a_i$ and value $v_i$.

\subsection{Concurrent Rewriting Logic}

Concurrent rewriting is a formal model that allows one to reason about concurrent systems in terms of state transitions. A key strength of concurrent rewriting is its ability to represent many models of concurrency within a common framework.

A rewriting logic theory is composed of a \emph{signature} denoted as $(\Sigma, E)$, together with a set of rewrite theory rules. Where $\Sigma$ defines the function symbols of the theory and the set $E$ consists of $\Sigma$-equations. A $\Sigma$-algebra is a set $V$ where each $f\in \Sigma$ of $n$ arguments is associated with a function $f: V^n\rightarrow V$. The Symbol $T_{\Sigma}$ denotes the $\Sigma$-algebra of ground $\Sigma$-terms. The symbol $T_{\Sigma, E}$ denotes the $\Sigma$-algebra of equivalence classes of ground $\Sigma$-terms modulo the equation $E$. If $X$ denotes a countable set of variables, then the symbols $T_{\Sigma}(X)$ and $T_{\Sigma, E}(X)$ denote the $\Sigma$-algebra of $\Sigma$-terms with variables in $X$, and the $\Sigma$-algebra of equivalence classes of $\Sigma$-terms with variables in $X$ modulo the equations $E$. Given a term $t\in T_{\Sigma}(\{x_1,...,x_n\})$ and term $u_1,...,u_n$, $t(u_1/x_1,...,u_n/x_n)$ denotes the term obtained from $t$ by simultaneously substituting each $u_i$ for each $x_i$. And for a term $t$, $[t]$ denotes the $E$-equivalence class of $t$.

Given a signature $(\Sigma, E)$, the rewrite rules of a theory are represented by a pair $(L, R)$ where $L$ is a set of \emph{labels} and a rule is denoted as $r: [t] \rightarrow [t']$. Then the 4-tuple $\mathbb{R} = (\Sigma, E, L, R)$ is called a rewrite theory. The deduction rules establishing a rewrite theory logic are defined as follows.

\begin{itemize}
  \item \textbf{Reflexivity}: For each $[t]\in T_{\Sigma, E}(X)$,

        \begin{center}
            $$\frac{ }{[t]\rightarrow [t]}$$
        \end{center}
  \item \textbf{Congruence}: For each $f\in \Sigma$ of $n\in \mathbb{N}$,

        \begin{center}
            $$\frac{[t_1]\rightarrow[t'_1]...[t_n]\rightarrow[t'_n]}{f(t_1,...,t_n)\rightarrow f(t'_1,...,t'_n)}$$
        \end{center}
  \item \textbf{Replacement}: For each rule $r:[t(x_1,...,x_n)]\rightarrow [t'(x_1,...,x_n)]$,

        \begin{center}
            $$\frac{[w_1]\rightarrow [w'_1]...[w_n]\rightarrow [w'_n]}{[t(\overline{w}/\overline{x})]\rightarrow [t'(\overline{w'}/\overline{x})]}$$
        \end{center}
  \item \textbf{Transitivity}:

        \begin{center}
            $$\frac{[t_1]\rightarrow [t_2] [t_2]\rightarrow[t_3]}{[t_1]\rightarrow [t_3]}$$
        \end{center}
\end{itemize}

\subsection{Rewriting Semantics for Actors}

In \cite{ActorCustandCompo}, a rewriting semantics based on abstract actor structure (AAS)\cite{ActorRewriting}\cite{ActorRewriting1} for actors shown in Fig.\ref{Fig.4} was introduced. We introduce some notions and definitions here again as a basis of the following sections. Note that in \cite{ActorCustandCompo} and the following sections only \emph{fair}\footnote[1]{exactly observation fairness, see in \cite{ActorRewriting}} rewriting derivations are considered.

An actor is modeled a 6-tuples $(p, a, b, s, l, t)$ with attributes introduced in the above section. And the term constructors follows.

$\_(\_,\_|\sigma:[\_]\lambda:[\_]\tau:[\_]):(\mathbb{P}\times\mathbb{A}\times\mathbb{B}\times\mathbb{S}\times\mathbb{E}\times\mathbb{A})\xrightarrow{\circ}\mathbb{F}$
denotes an actor.

$\_\blacktriangleleft(\_,\_,\_):\mathbb{A}\times\mathbb{A}\times\mathbb{E}\times\mathbb{V}\rightarrow\mathbb{F}$
denotes an event message. In a event message $r\blacktriangleleft(s,e,v)$, $r$ is the destination address, $s$ is the source address, $e$ is the event and $v$ is the values of the event.

$\_:\_\lhd\_:\mathbb{A}\times\mathbb{A}\times\mathbb{V}\rightarrow\mathbb{F}$ denotes a message $a:a'\lhd v$ with values $v$ sent from an actor with address $a'$ to an actor with address $a$.

$\diamond \mathbb{F}$ denotes the empty fragment.

$\_,\_:\mathbb{F}\times\mathbb{F}\xrightarrow{\circ}\mathbb{F}$ denotes composition of fragment. It is true that $F_1, F_2\in \mathbb{F}$ if $recep(F_1)\cap recep(F_2)=\emptyset$.

$[\_]\lceil \{\_\}:\mathbb{F}\times P_\omega[\mathbb{A}]\xrightarrow{\circ}\mathbb{F}$ denotes a restriction of a fragment. $[F]\lceil \{a_1,...,a_n\}\in \mathbb{F}$ if $\{a_1,...,a_n\}\subset recep(F)$.

$\langle\_\rangle:\mathbb{F}\xrightarrow{\circ}\mathbb{C}$ denotes a configuration.

The relations to operate actor configurations are as following.

$block\subset(\mathbb{E}\times\mathbb{E})$ defines the block relation of two events.

$newEvent:(\mathbb{S}\times\mathbb{B})\xrightarrow{\circ}(\mathbb{E}\times\mathbb{V})$ defines the new event and its value generated by a running actor with specific state and behavior.

$nextState:(\mathbb{S}\times\mathbb{B}\times\mathbb{E}\times\mathbb{V})\xrightarrow{\circ}\mathbb{S}$ defines the new state of an actor.

$newInstance:(\mathbb{B}\times\mathbb{V})\xrightarrow{\circ}\mathbb{S}$ defines a new actor with initial state.

$recep:\mathbb{F}\rightarrow P_\omega[\mathbb{A}]$ defines receptionists of a fragment.

$acq:(\mathbb{S}\times\mathbb{V})\rightarrow P_\omega[\mathbb{A}]$ defines acquaintances of a state or a value.

$extern:\mathbb{F}\rightarrow P_\omega[\mathbb{A}]$ defines external actors of a fragment.

$\widehat{\_}:\mathbb{A}\xrightarrow[1-1]{}\mathbb{A}\rightarrow\mathbb{S}\cup\mathbb{V}\cup\mathbb{F}\xrightarrow[1-1]{}\mathbb{S}\cup\mathbb{V}\cup\mathbb{F}$
defines a renaming function.

$En_d\subseteq \mathbb{A}\times\mathbb{S}\times\mathbb{V}$ defines a predicate on actors and values, and $En_d(a,s,v)$ holds if an actor with an address $a$ and state $s$ is enabled for deliver a message with value $v$.

$En_{ex}\subseteq \mathbb{A}\times\mathbb{S}$ defines a predicate on actors, and $En_{ex}(a,s)$ holds if an actor with an a address $a$ and state $s$ is enabled for execution.

The axioms (called hygiene conditions in \cite{ActorCustandCompo}) about the above term constructors and relations over actor configurations will no longer be introduced, please see \cite{ActorRewriting} and \cite{ActorCustandCompo}.

We give the basic computation steps with \emph{block} relations of events in Fig.\ref{Fig.4} as follow, which was introduced in \cite{ActorCustandCompo}.

[request]

$!(a,b|\sigma:[s]\lambda:[l]\tau:[a'])\rightarrow ?(a,b|\sigma:[s]\lambda:[e]\tau:[a']), a'\blacktriangleleft(a,e,v)$

if $nextEvent(s,b)=(e,v)$.

The \textbf{request} rule transforms a running actor into a ready actor and an event pair. The $nextEvent$ function determines the next signal generate by an actor based on its behaviors and current state. The transition map of an actor determines where the new signal should be processed. After generating a signal, an actor is blocked until an appropriate notification event is received.

[compute]

$?(a,b|\sigma:[s]\lambda:[l]\tau:[a']),a\blacktriangleleft(a',e,v)\rightarrow !(a,b|\sigma:[s']\lambda:[l]\tau:[a'])$

if $(l,e)\in block$ and $newState(s,b,e,v)=s'$.

The \textbf{compute} rule allows a blocked actor to become active by processing a notification event. The $newState$ function determines the new state of an actor after it has processed a notification. This rule abstracts over the internal computation performed by an actor upon receiving a particular notification.

The rewriting semantics of the common actor with block relations of events will be omitted here, please see \cite{ActorCustandCompo}.

\subsection{A Semantics for AA and WSO}\label{AA and WSO Semantics}

Firstly, we do not concern the semantics of aspects of interactions between AAs and applications inside. There are at least two reasons: (1)They have similar behaviors with those interacting with WSs and different low level communication mechanisms, such as local object method call, RPC, DCOM, RMI, etc, that do not affect their semantics. (2)This paper focuses on WS Composition.

An AA interacts with other AAs and WSs outside via its WSO and can not create other AAs. The event sort $\mathbb{E}$ including signals and notifications in Fig.\ref{Fig.4} as follows.

$\mathbb{E}=\{\textbf{transmit}, \textbf{ready}, \textbf{complete}, \textbf{deliver}\}$

An AA is defined as a 8-tuples $(p,a,b,s,l,t,wso,ws)$. $p,a,b,s,l,t$ denotes as the above section showing, $wso$ denotes AA's WSO, and $ws$ denotes WS of the WSO. An AA has the form of $p(a,b,wso,ws|\sigma:[s]\lambda:[l]\tau:[t])$ as a term.

The semantics of an AA is defined as the following rules.

[\textbf{send-in}]
\\

$A,A',WSO\blacktriangleleft(a,\textbf{transmit},\{a',v\})\rightarrow A,A',a\blacktriangleleft(WSO,\textbf{complete},\{\}), a'\lhd v$

\quad\quad\quad\quad if $A=p_a(a,b_a,wso_a,ws_a|\sigma:[s_a]\lambda:[l_a]\tau:[t_a])$ and

\quad\quad\quad\quad $A'=p_{a'}(a',b_{a'},wso_{a'},ws_{a'}|\sigma:[s_{a'}]\lambda:[l_{a'}]\tau:[t_{a'}])$ and

\quad\quad\quad\quad $wso_a = wso_{a'}=WSO$ and $ws_a = ws_{a'}$.
\\

This \textbf{send-in} rule illustrates that only AAs within a same WSO can exchange between each other via its WSO through a \textbf{transmit}-\textbf{complete} event pair.

[\textbf{send-out}]
\\

$\langle F,A,WSO\blacktriangleleft(a,\textbf{transmit},\{a',v\})\rightarrow F,A,a\blacktriangleleft(WSO,\textbf{complete},\{\}), a'\lhd v\rangle$

\quad\quad\quad\quad if $A=p_a(a,b_a,wso_a,ws_a|\sigma:[s_a]\lambda:[l_a]\tau:[t_a])$ and

\quad\quad\quad\quad $a'\notin recep(F)$ and $wso_a=WSO$.
\\

This \textbf{send-out} rule allows an AA within a WSO to send a message to applications outside via its WSO also through a \textbf{transmit}-\textbf{complete} event pair.

[\textbf{ready}]
\\

$WSO\blacktriangleleft(a,\textbf{ready},\{\}), a:a'\lhd v \rightarrow a\blacktriangleleft(WSO,\textbf{deliver},\{v\})$.
\\

This \textbf{ready} rule illustrates that an AA in ready can be deliver a message by its WSO through a \textbf{ready}-\textbf{deliver} event pair.

[\textbf{in}]
\\

$\langle F\rangle\rightarrow \langle F, a\lhd v\rangle$

\quad\quad\quad\quad if $F=[F', A]\lceil\{R\}$ and

\quad\quad\quad\quad $A=p(a,b,wso,ws|\sigma:[s]\lambda:[l]\tau:[t])$ and $a=wso$ and $a\in recep(F)$.
\\

The \textbf{in} rule shows that only messages sent to the WSO from applications outside are allowed in.

[\textbf{out}]
\\

$\langle [F,a:a'\lhd v]\rangle\lceil\{R\}\rightarrow \langle [F]\lceil\{R\cup (acq(v)\cap recep(F))\}\rangle$

\quad\quad\quad\quad if $a\notin recep(F)$.
\\

The \textbf{out} rule allows to send a message to the exterior.

As follows, we give the \textbf{create-AA} rule of the WSO.

[\textbf{create-AA}]
\\

$A\rightarrow [A, A', WSO\blacktriangleleft(a',\textbf{ready},\{\})]\lceil\{a\}$

\quad\quad\quad\quad if $A=?(a,b_a,wso_a,ws_a|\sigma:[s_a]\lambda:[l_a]\tau:[t_a])$ and

\quad\quad\quad\quad $A'=?(a',b_{a'},wso_{a'},ws_{a'}|\sigma:[s_{a'}]\lambda:[\textbf{ready}]\tau:[WSO])$ and

\quad\quad\quad\quad $En_{ex}(a,s_a)$ holds and

\quad\quad\quad\quad $a'\notin acq(s)\cup\{a\}$

\quad\quad\quad\quad and $a=wso_a=wso_{a'}=WSO$ and $ws_a=ws_{a'}$.
\\

The \textbf{create-AA} rule allows that a WSO creates its AAs. The created AA has the same $wso$ and $ws$ attributes to those of the WSO.

\subsection{A Semantics for WS}\label{WS Semantics}

A WS is denoted as a 8-tuples $p(a,b,wso,ws|\sigma:[s]\lambda:[l]\tau:[t])$ where $wso$ denotes its WSO created and $ws$ denotes its partner WS.

The functions of a WS includes creation of WSOs and interaction with its WSO and partner WS. So, we give the \textbf{create-WSO} rule, \textbf{in} rule and \textbf{out} rule below.

[\textbf{create-WSO}]
\\

$A\rightarrow [A^*, A', a\blacktriangleleft(a',\textbf{ready},\{\})]\lceil\{a\}$

\quad\quad\quad\quad if $A=?(a,b_a,[\quad],ws_a|\sigma:[s_a]\lambda:[l_a]\tau:[t_a])$ and

\quad\quad\quad\quad $A^*=?(a,b_a,a',ws_a|\sigma:[s_a]\lambda:[l_a]\tau:[t_a])$ and

\quad\quad\quad\quad $A'=?(a',b_{a'},a',a|\sigma:[s_{a'}]\lambda:[\textbf{ready}]\tau:[a])$ and

\quad\quad\quad\quad $En_{ex}(a,s_a)$ holds and

\quad\quad\quad\quad $a'\notin acq(s)\cup\{a\}$.
\\

The \textbf{create-WSO} rule allows a WS to create its WSO and establishes their relation at the same time.

[\textbf{in}]
\\

$\langle F\rangle\rightarrow \langle F, a\lhd v\rangle$

\quad\quad\quad\quad if $F=[F', A]\lceil\{R\}$ and

\quad\quad\quad\quad $A=p(a,b,wso,ws|\sigma:[s]\lambda:[l]\tau:[t])$ and $a\in recep(F)$.
\\

The \textbf{in} rule shows that messages sent to the WS from its WSO or its partner WS are allowed in.

[\textbf{out}]
\\

$\langle [F,a:a'\lhd v]\rangle\lceil\{R\}\rightarrow \langle [F]\lceil\{R\cup (acq(v)\cap recep(F))\}\rangle$

\quad\quad\quad\quad if $a\notin recep(F)$.
\\

The \textbf{out} rule allows to send a message to the exterior.

\subsection{A Semantics for WSC}

A WSC is also denoted as a 8-tuples $p(a,b,ws1,ws2|\sigma:[s]\lambda:[l]\tau:[t])$ where $ws1$ and $ws2$ denote the partner WSs created.

The main functions of a WSC are creation of partner WSs. So, we give the \textbf{create-WSs} rule below.

[\textbf{create-WSs}]
\\

$A\rightarrow [A^*, A', A'']\lceil\{a\}$

\quad\quad\quad\quad if $A=?(a,b_a,[\quad],[\quad]|\sigma:[s_a]\lambda:[l_a]\tau:[t_a])$ and

\quad\quad\quad\quad $A^*=?(a,b_a,a',a''|\sigma:[s_a]\lambda:[l_a]\tau:[t_a])$ and

\quad\quad\quad\quad $A'=?(a',b_{a'},[\quad],a''|\sigma:[s_{a'}]\lambda:[l_{a'}]\tau:[t_{a'}])$ and

\quad\quad\quad\quad $A''=?(a'',b_{a''},[\quad],a'|\sigma:[s_{a''}]\lambda:[l_{a''}]\tau:[t_{a''}])$ and

\quad\quad\quad\quad $En_{ex}(a,s_a)$ holds and

\quad\quad\quad\quad $a'\notin acq(s)\cup acq(s_{a''})\cup\{a\}$ and $a''\notin acq(s)\cup acq(s_{a'})\cup\{a\}$.
\\

The \textbf{create-WSs} rule allows a WSC to create its partner WSs and establishes their relation at the same time.

\subsection{Interaction Semantics}

In this section, we will discuss the formal relationships between WSO and WS, WS and WS, further more, WSO and WSO. Through analysis of interaction semantics among them, we establish the relationship among them by term of \emph{compositionality}. The term compositionality in interaction semantics between WSO and WS means that the messages sent by AAs in a WSO to actors in the interface WS correspond to the messages expected by the actors in the WS, and vice versa. And compositionality in interactions semantics between WS and WS, WSO and WSO is similar to that between WSO and WS.

In \cite{ActorRewriting1}, \cite{ActorRewriting2}and \cite{ActorCustandCompo}, interaction semantics between distributed components are analyzed. We use the notions and definitions about interaction semantics in \cite{ActorRewriting1}, \cite{ActorRewriting2} and \cite{ActorCustandCompo} as a basis of our further conclusions, and we do not give explanations more.

\subsubsection{Interaction Semantics between WSO and WS}

To model a WSO configuration, we introduce an operator:

$\langle \_\rangle[\_]:\mathbb{F}\times(\mathbb{A}\times\mathbb{A})\xrightarrow{\circ}\mathbb{C}$

with $\langle A\rangle[b]$ where $recep(A)\cap \{b^*\} = \emptyset$ where $(a^*,b^*)=b$.

$A$ denotes an actor fragment, $b$ denotes the relation of a WSO and its WS pair. We use a members function defined as $members:\mathbb{F}\rightarrow P_\omega(\mathbb{A})$.

We use $g$ to denote the collection of rules in section \ref{AA and WSO Semantics}, excluding the \textbf{in} and \textbf{out} rules. The set of interaction steps $I_g$ with WS is defined as follows.

[\textbf{wso-ws-silent}]
\\

$\langle A\rangle[b]\Rightarrow\langle A'\rangle[b]$

\quad\quad\quad\quad if $A\xrightarrow{g}A'$.
\\

[\textbf{wso-ws-emit-1(a,a',e,v)}]
\\

$\langle A,a'\blacktriangleleft(a,e,v)\rangle[b]\Rightarrow\langle A\rangle[b]$

\quad\quad\quad\quad if $a'\notin members(A)$.
\\

[\textbf{wso-ws-emit-2(a,a',v)}]
\\

$\langle A,a':a\lhd v\rangle[b]\Rightarrow\langle A\rangle[b]$

\quad\quad\quad\quad if $a'\notin members(A)$.
\\

[\textbf{wso-ws-consume-1(a',a,e,v)}]
\\

$\langle A\rangle[b]\Rightarrow\langle A,a\blacktriangleleft(a',e,v)\rangle[b]$

\quad\quad\quad\quad if $a\in members(A)$ and $(recep(A),a')=b$.
\\

[\textbf{wso-ws-consume-2(a',a,v)}]
\\

$\langle A\rangle[b]\Rightarrow\langle A,a:a'\lhd v\rangle[b]$

\quad\quad\quad\quad if $a\in members(A)$ and $(recep(A),a')=b$.
\\

To model a WS configuration, we also introduce another operator:

$\langle\_\rangle[\_]:\mathbb{F}\times(\mathbb{A}\times\mathbb{A})\xrightarrow{\circ}\mathbb{C}$

with $\langle M\rangle[b]$ where $recep(M)\cap \{b^*\} = \emptyset$ where $(a^*,b^*)=b$.

$M$ denotes an actor fragment, $b$ denotes the relation of a WSO and its WS.

We use $g$ to denote the collection of rules in section \ref{WS Semantics}, excluding the \textbf{in} and \textbf{out} rules. The set of interaction steps $I_g$ with WSO is defined as follows.

[\textbf{ws-wso-silent}]
\\

$\langle M\rangle[b]\Rightarrow\langle M'\rangle[b]$

\quad\quad\quad\quad if $M\xrightarrow{g}M'$.
\\

[\textbf{ws-wso-emit-1(a,a',e,v)}]
\\

$\langle M,a'\blacktriangleleft(a,e,v)\rangle[b]\Rightarrow\langle M\rangle[b]$

\quad\quad\quad\quad if $(a,a')= b$.
\\

[\textbf{ws-wso-emit-2(a,a',v)}]
\\

$\langle M,a':a\lhd v\rangle[b]\Rightarrow\langle M\rangle[b]$

\quad\quad\quad\quad if $(a,a')= b$.
\\

[\textbf{ws-wso-consume-1(a',a,e,v)}]
\\

$\langle M\rangle[b]\Rightarrow\langle M,a\blacktriangleleft(a',e,v)\rangle[b]$

\quad\quad\quad\quad if $(a,a')= b$.
\\

[\textbf{ws-wso-consume-2(a',a,v)}]
\\

$\langle M\rangle[b]\Rightarrow\langle M,a:a'\lhd v\rangle[b]$

\quad\quad\quad\quad if $(a,a')= b$.
\\

We define the interaction sequence $\zeta(i)$ and its sequence dual $\mathcal{D}(\zeta)=\zeta'(i)$ between WSO and WS as follow:

(1) If $\zeta(i)=\textbf{wso-ws-silent}$, then $\zeta'(i)=\textbf{ws-wso-silent}$.

(2) If $\zeta(i)=\textbf{wso-ws-emit-1(a,a',e,v)}$, then $\zeta'(i)=\textbf{ws-wso-consume-1(a',a,e,v)}$.

(3) If $\zeta(i)=\textbf{wso-ws-emit-2(a,a',v)}$, then $\zeta'(i)=\textbf{ws-wso-consume-2(a',a,v)}$.

(4) If $\zeta(i)=\textbf{wso-ws-consume-1(a',a,e,v)}$, then $\zeta'(i)=\textbf{ws-wso-emit-1(a,a',e,v)}$.

(5) If $\zeta(i)=\textbf{wso-ws-consume-2(a',a,v)}$, then $\zeta'(i)=\textbf{ws-wso-emit-2(a,a',v)}$.

Then we can get the interaction semantics of $\langle A\rangle[b]$ and $\langle M\rangle[b]$ as $\mathcal{I}(\langle A\rangle[b])$ and $\mathcal{I}(\langle M\rangle[b])$. The so-called configuration images $cimage_{(\langle A\rangle[b])}(\langle M\rangle[b])$ and $cimage_{(\langle M\rangle[b])}(\langle A\rangle[b])$ can be established by use of interaction semantics $\mathcal{I}(\langle A\rangle[b])$ and $\mathcal{I}(\langle M\rangle[b])$, and also interaction sequence $\zeta(i)$ and its sequence dual $\zeta'(i)$ defined above. And then we can define the compatibility among the configuration $\langle A\rangle[b]$ and the configuration $\langle M\rangle[b]$ like that in \cite{ActorCustandCompo}.

Finally, we can draw our conclusions as follows according to the definition of compositionality.

\begin{theorem}
(\textbf{Compositionality between WSO and WS}) The partial configuration $\langle A\rangle[b]$ is \emph{composable} with the partial configuration $\langle M\rangle[b]$ just if $members(A) \cap members(M)=\emptyset$, and $\langle A\rangle[b]$ is compatible with $\langle A\rangle[b]$ (and vice versa).
\end{theorem}

\begin{proof}
See above.
\end{proof}

\subsubsection{Interaction Semantics between WS and WS}

To model a WS configuration, we still use the operator:

$\langle\_\rangle[\_]:\mathbb{F}\times(\mathbb{A}\times\mathbb{A})\xrightarrow{\circ}\mathbb{C}$

with $\langle M\rangle[d]$ where $recep(M)\cap \{d^*\} = \emptyset$ where $(a^*,d^*)=d$.

$M$ still denotes an actor fragment, but $d$ denotes the relation of a WS and its partner WS.

We also use $g$ to denote the collection of rules in section \ref{WS Semantics}, excluding the \textbf{in} and \textbf{out} rules. The set of interaction steps $I_g$ with WS is defined as follows.

[\textbf{ws-ws-silent}]
\\

$\langle M\rangle[d]\Rightarrow\langle M'\rangle[d]$

\quad\quad\quad\quad if $M\xrightarrow{g}M'$.
\\

[\textbf{ws-ws-emit-1(a,a',e,v)}]
\\

$\langle M,a'\blacktriangleleft(a,e,v)\rangle[d]\Rightarrow\langle M\rangle[d]$

\quad\quad\quad\quad if $(a,a')= d$.
\\

[\textbf{ws-ws-emit-2(a,a',v)}]
\\

$\langle M,a':a\lhd v\rangle[d]\Rightarrow\langle M\rangle[d]$

\quad\quad\quad\quad if $(a,a')= d$.
\\

[\textbf{ws-ws-consume-1(a',a,e,v)}]
\\

$\langle M\rangle[d]\Rightarrow\langle M,a\blacktriangleleft(a',e,v)\rangle[d]$

\quad\quad\quad\quad if $(a,a')= d$.
\\

[\textbf{ws-ws-consume-2(a',a,v)}]
\\

$\langle M\rangle[d]\Rightarrow\langle M,a:a'\lhd v\rangle[d]$

\quad\quad\quad\quad if $(a,a')= d$.
\\

We define the interaction sequence $\zeta(i)$ and its sequence dual $\zeta'(i)$ between WS and WS as follow:

(1) If $\zeta(i)=\textbf{ws-ws-silent}$, then $\zeta'(i)=\textbf{ws-ws-silent}$.

(2) If $\zeta(i)=\textbf{ws-ws-emit-1(a,a',e,v)}$, then $\zeta'(i)=\textbf{ws-ws-consume-1(a',a,e,v)}$.

(3) If $\zeta(i)=\textbf{ws-ws-emit-2(a,a',v)}$, then $\zeta'(i)=\textbf{ws-ws-consume-2(a',a,v)}$.

(4) If $\zeta(i)=\textbf{ws-ws-consume-1(a',a,e,v)}$, then $\zeta'(i)=\textbf{ws-ws-emit-1(a,a',e,v)}$.

(5) If $\zeta(i)=\textbf{ws-ws-consume-2(a',a,v)}$, then $\zeta'(i)=\textbf{ws-ws-emit-2(a,a',v)}$.

Similar to interactions between WSO and WS, we get the following conclusions.

\begin{theorem}
(\textbf{Compositionality between WS and WS}) The partial configuration $\langle M\rangle[b]$ is \emph{composable} with the partial configuration $\langle M'\rangle[b']$ just if $members(M) \cap members(M')=\emptyset$, and $\langle M\rangle[b]$ is compatible with $\langle M'\rangle[b']$ (and vice versa).
\end{theorem}

\subsubsection{Interaction Semantics between WSO and WSO}

Though a WSO does interact with another WSO directly, in fact there are virtual interactions between WSO and WSO omitting the middle interactions -- WSO and WS, WS and WS, WS and WSO.

To model a WSO configuration, we introduce a new operator:

$\langle \_\rangle (\_):\mathbb{F}\times \mathbb{A}\xrightarrow{\circ}\mathbb{C}$

with $\langle A\rangle(w)$ where $recep(A)\cap \{w\}=\emptyset$.

$A$ denotes an actor fragment, and $w$ denotes the external actor of the configuration, such as a WS. We also use a members function defined as $members:\mathbb{F}\rightarrow P_\omega(\mathbb{A})$.

We use $g$ to denote the collection of rules in section \ref{AA and WSO Semantics}, excluding the \textbf{in} and \textbf{out} rules. The set of interaction steps $I_g$ is defined as follows.

[\textbf{wso-wso-silent}]
\\

$\langle A\rangle (w)\Rightarrow\langle A'\rangle (w)$

\quad\quad\quad\quad if $A\xrightarrow{g}A'$.
\\

[\textbf{wso-wso-emit-1(a,a',e,v)}]
\\

$\langle A,a'\blacktriangleleft(a,e,v)\rangle(w)\Rightarrow\langle A\rangle(w)$

\quad\quad\quad\quad if $a'\notin members(A)$.
\\

[\textbf{wso-wso-emit-2(a,a',v)}]
\\

$\langle A,a':a\lhd v\rangle(w)\Rightarrow\langle A\rangle(w)$

\quad\quad\quad\quad if $a'\notin members(A)$.
\\

[\textbf{wso-wso-consume-1(a',a,e,v)}]
\\

$\langle A\rangle(w)\Rightarrow\langle A,a\blacktriangleleft(a',e,v)\rangle(w)$

\quad\quad\quad\quad if $a\in members(A)$ and $a'=w$.
\\

[\textbf{wso-wso-consume-2(a',a,v)}]
\\

$\langle A\rangle(w)\Rightarrow\langle A,a:a'\lhd v\rangle(w)$

\quad\quad\quad\quad if $a\in members(A)$ and $a'=w$.
\\

We define the interaction sequence $\zeta(i)$ and its sequence dual $\zeta'(i)$ between WSO and WSO as follow:

(1) If $\zeta(i)=\textbf{wso-wso-silent}$, then $\zeta'(i)=\textbf{wso-wso-silent}$.

(2) If $\zeta(i)=\textbf{wso-wso-emit-1(a,a',e,v)}$, then $\zeta'(i)=\textbf{wso-wso-consume-1(a',a,e,v)}$.

(3) If $\zeta(i)=\textbf{wso-wso-emit-2(a,a',v)}$, then $\zeta'(i)=\textbf{wso-wso-consume-2(a',a,v)}$.

(4) If $\zeta(i)=\textbf{wso-wso-consume-1(a',a,e,v)}$, then $\zeta'(i)=\textbf{wso-wso-emit-1(a,a',e,v)}$.

(5) If $\zeta(i)=\textbf{wso-wso-consume-2(a',a,v)}$, then $\zeta'(i)=\textbf{wso-wso-emit-2(a,a',v)}$.

Similarly, we get the following conclusions.

\begin{theorem}
(\textbf{Compositionality between WSO and WSO}) The partial configuration $\langle A\rangle(w)$ is \emph{composable} with the partial configuration $\langle A'\rangle(w')$ just if $members(A) \cap members(A')=\emptyset$, and $\langle A\rangle(w)$ is compatible with $\langle A'\rangle(w')$ (and vice versa).
\end{theorem}

Note that the compositionality between WSO and WSO can also be gotten from the the compositionality between WSO and WS, and the compositionality between WS and WS.

\subsection{Semantics for Buying Books Example}

Semantics for Buying Books example includes semantics for UserAgentWSO and its AAs, semantics for BookStoreWSO and its AAs, semantics for UserAgentWS, semantics for BookStoreWS, semantics for BuyingBookWSC, interaction semantics between UserAgentWSO and UserAgentWS, interaction semantics between UserAgentWS and BookStoreWS, interaction semantics between BookStoreWS and BookStoreWSO and interactions between UserAgentWSO and BookStoreWSO.

The detailed semantics for Buying Books Example is following. We also pay attentions to the UserAgent side, that is, semantics for UserAgentWSO and its AAs, semantics for UserAgentWS, semantics for BuyingBookWSC, interaction semantics between UserAgentWSO and UserAgentWS, interaction semantics between UserAgentWS and BookStoreWS, and interactions between UserAgentWSO and BookStoreWSO are as following.

Semantics for UserAgentWSO and its AAs:

[\textbf{in}]
\\

$\langle F\rangle\rightarrow \langle F, a\lhd v\rangle$

\quad\quad\quad\quad if $F=[F', A]\lceil\{R\}$ and

\quad\quad\quad\quad $A=p(a,b,UserAgentWSO,ws|\sigma:[s]\lambda:[l]\tau:[t])$ and $a=UserAgentWSO$ and $a\in recep(F)$.
\\

[\textbf{out}]
\\

$\langle [F,a:a'\lhd v]\rangle\lceil\{R\}\rightarrow \langle [F]\lceil\{R\cup (acq(v)\cap recep(F))\}\rangle$

\quad\quad\quad\quad if $a\notin recep(F)$.
\\

[\textbf{create-RequstLBAA}]
\\

$A\rightarrow [A, A', WSO\blacktriangleleft(a',\textbf{ready},\{\})]\lceil\{a\}$

\quad\quad\quad\quad if $A=?(a,b_a,wso_a,ws_a|\sigma:[s_a]\lambda:[l_a]\tau:[t_a])$ and

\quad\quad\quad\quad $A'=?(a',b_{a'},wso_{a'},ws_{a'}|\sigma:[s_{a'}]\lambda:[\textbf{ready}]\tau:[WSO])$ and

\quad\quad\quad\quad $En_{ex}(a,s_a)$ holds and

\quad\quad\quad\quad $a'\notin acq(s)\cup\{a\}$

\quad\quad\quad\quad and $a=wso_a=wso_{a'}=UserAgentWSO$ and $ws_a=ws_{a'}$ and $a'=RequstLBAA$.
\\

Similarly to \textbf{create-RequstLBAA} rule, there are also \textbf{create-ReceiveLBAA} rule, \textbf{create-SendSBAA} rule, \textbf{create-ReceivePBAA} rule, and \textbf{create-PayBAA} rule.

Semantics for UserAgentWS:

[\textbf{create-UserAgentWSO}]
\\

$A\rightarrow [A^*, A', a\blacktriangleleft(a',\textbf{ready},\{\})]\lceil\{a\}$

\quad\quad\quad\quad if $A=?(a,b_a,[\quad],ws_a|\sigma:[s_a]\lambda:[l_a]\tau:[t_a])$ and

\quad\quad\quad\quad $A^*=?(a,b_a,a',ws_a|\sigma:[s_a]\lambda:[l_a]\tau:[t_a])$ and

\quad\quad\quad\quad $A'=?(a',b_{a'},a',a|\sigma:[s_{a'}]\lambda:[\textbf{ready}]\tau:[a])$ and

\quad\quad\quad\quad $En_{ex}(a,s_a)$ holds and

\quad\quad\quad\quad $a'\notin acq(s)\cup\{a\}$ and $a=UserAgentWS$ and $a'=UserAgentWSO$.
\\

[\textbf{in}]
\\

$\langle F\rangle\rightarrow \langle F, a\lhd v\rangle$

\quad\quad\quad\quad if $F=[F', A]\lceil\{R\}$ and

\quad\quad\quad\quad $A=p(a,b,UserAgentWSO,UserAgentWS|\sigma:[s]\lambda:[l]\tau:[t])$ and $a\in recep(F)$.
\\

[\textbf{out}]
\\

$\langle [F,a:a'\lhd v]\rangle\lceil\{R\}\rightarrow \langle [F]\lceil\{R\cup (acq(v)\cap recep(F))\}\rangle$

\quad\quad\quad\quad if $a\notin recep(F)$.
\\

Semantics for BuyingBookWSC:

[\textbf{create-WSs}]
\\

$A\rightarrow [A^*, A', A'']\lceil\{a\}$

\quad\quad\quad\quad if $A=?(a,b_a,[\quad],[\quad]|\sigma:[s_a]\lambda:[l_a]\tau:[t_a])$ and

\quad\quad\quad\quad $A^*=?(a,b_a,a',a''|\sigma:[s_a]\lambda:[l_a]\tau:[t_a])$ and

\quad\quad\quad\quad $A'=?(a',b_{a'},[\quad],a''|\sigma:[s_{a'}]\lambda:[l_{a'}]\tau:[t_{a'}])$ and

\quad\quad\quad\quad $A''=?(a'',b_{a''},[\quad],a'|\sigma:[s_{a''}]\lambda:[l_{a''}]\tau:[t_{a''}])$ and

\quad\quad\quad\quad $En_{ex}(a,s_a)$ holds and

\quad\quad\quad\quad $a'\notin acq(s)\cup acq(s_{a''})\cup\{a\}$ and $a''\notin acq(s)\cup acq(s_{a'})\cup\{a\}$ and $a'=UserAgentWS$ and $a''=BookStoreWS$.
\\

Interaction semantics between UserAgentWSO and UserAgentWS:

(1) At UserAgentWSO side, $\textbf{wso-ws-emit-2(UserAgentWS,UserAgentWSO,requestLB)}$, then $\textbf{ws-wso-consume-2(UserAgentWSO,UserAgentWS,requestLB)}$ at UserAgentWS side.

(2) At UserAgentWS side, $\textbf{ws-wso-emit-2(UserAgentWSO,UserAgentWS,receiveLB)}$, then $\textbf{wso-ws-consume-2(UserAgentWS,UserAgentWSO,receiveLB)}$ at UserAgentWSO side.

(3) At UserAgentWSO side, $\textbf{wso-ws-emit-2(UserAgentWS,UserAgentWSO,sendSB)}$, then $\textbf{ws-wso-consume-2(UserAgentWSO,UserAgentWS,sendSB)}$ at UserAgentWS side.

(4) At UserAgentWS side, $\textbf{ws-wso-emit-2(UserAgentWSO,UserAgentWS,receivePB)}$, then $\textbf{wso-ws-consume-2(UserAgentWS,UserAgentWSO,receivePB)}$ at UserAgentWSO side.

(5) At UserAgentWSO side, $\textbf{wso-ws-emit-2(UserAgentWS,UserAgentWSO,payB)}$, then $\textbf{ws-wso-consume-2(UserAgentWSO,UserAgentWS,payB)}$ at UserAgentWS side.
\\

We can see that the partial configurations of UserAgentWSO and UserAgentWS are composable.

Interaction semantics between UserAgentWS and BookStoreWS:

(1) At UserAgentWS side, $\textbf{ws-ws-emit-2(UserAgentWS,BookStoreWS,requestLB)}$, then $\textbf{ws-ws-consume-2(BookStoreWS,UserAgentWS,requestLB)}$ at BookStoreWS side.

(2) At BookStoreWS side, $\textbf{ws-ws-emit-2(BookStoreWS,UserAgentWS,receiveLB)}$, then $\textbf{ws-ws-consume-2(UserAgentWS,BookStoreWS,receiveLB)}$ at UserAgentWS side.

(3) At UserAgentWS side, $\textbf{ws-ws-emit-2(UserAgentWS,BookStoreWS,sendSB)}$, then $\textbf{ws-ws-consume-2(BookStoreWS,UserAgentWS,sendSB)}$ at BookStoreWS side.

(4) At BookStoreWS side, $\textbf{ws-ws-emit-2(UserAgentWSO,UserAgentWS,receivePB)}$, then $\textbf{ws-ws-consume-2(UserAgentWS,BookStoreWS,receivePB)}$ at UserAgentWS side.

(5) At UserAgentWS side, $\textbf{ws-ws-emit-2(UserAgentWS,UserAgentWSO,payB)}$, then $\textbf{ws-ws-consume-2(BookStoreWS,UserAgentWS,payB)}$ at BookStoreWS side.
\\

The partial configurations of UserAgentWS and BookStoreWS are composable.

Interactions between UserAgentWSO and BookStoreWSO:

(1) At UserAgentWSO side, $\textbf{wso-wso-emit-2(BookStoreWSO,UserAgentWSO,requestLB)}$, then $\textbf{wso-wso-consume-2(UserAgentWSO,BookStoreWSO,requestLB)}$ at BookStoreWSO side.

(2) At BookStoreWSO side, $\textbf{wso-wso-emit-2(UserAgentWSO,BookStoreWSO,receiveLB)}$, then $\textbf{wso-wso-consume-2(BookStoreWSO,UserAgentWSO,receiveLB)}$ at UserAgentWSO side.

(3) At UserAgentWSO side, $\textbf{wso-wso-emit-2(BookStoreWSO,UserAgentWSO,sendSB)}$, then $\textbf{wso-wso-consume-2(UserAgentWSO,BookStoreWSO,sendSB)}$ at BookStoreWSO side.

(4) At BookStoreWSO side, $\textbf{wso-wso-emit-2(UserAgentWSO,BookStoreWSO,receivePB)}$, then $\textbf{wso-wso-consume-2(BookStoreWSO,UserAgentWSO,receivePB)}$ at UserAgentWSO side.

(5) At UserAgentWSO side, $\textbf{wso-wso-emit-2(BookStoreWSO,UserAgentWSO,payB)}$, then $\textbf{wso-wso-consume-2(UserAgentWSO,BookStoreWSO,payB)}$ at BookStoreWSO side.

We can also see that the partial configurations of UserAgentWSO and BookStoreWSO are composable.

\section{Conclusions and Future Works}

We aim at the requirements of WS Composition and design an architecture of WS Composition runtime, which establishes a natural relationship between WSO and WSC. We introduce an actor-based language called AB-WSCL to support such an architecture. Finally, we give AB-WSCL a formal semantics to make it based on a firmly theoretic foundation.

Actually, AB-WSCL, its implementations and its formal semantics can be used widely. One main usage is in design time, as a modeling language, a simulation tool and a validator. The other is used as a real WS Composition runtime.

In this section, we firstly map ingredients of AB-WSCL to those of XML-based WS specifications to show the practicability of AB-WSCL. Then we conclude the advantages of AB-WSCL and point out the future works.

\subsection{Mappings between AB-WSCL and XML-Based WS Specifications}

The contents of AB-WSCL mainly involve three WS specifications, including WSDL\cite{WSDL}, WS-BPEL\cite{WS-BPEL} and WS-CDL\cite{WS-CDL}.

Mappings between AB-WSCL and WSDL are as Table \ref{Table.1} shows, AB-WSCL and WS-BPEL are as Table \ref{Table.2} illustrates, AB-WSCL and WS-CDL are as Table \ref{Table.3} illustrates.

\begin{center}
\begin{table}
  \begin{tabular}{@{}ll@{}}
   \hline
   WSDL\hspace{1cm}  & AB-WSCL \\
   \hline
   Web Service                 & WS Actor \\
   Message                     & local variable and WS-method args \\
   Operation                   & WS-method \\
   Operation Input             & WS-method args \\
   Operation Output            & sending messages by method call-back \\
   Message Exchange Pattern    & message deliver synchronously and asynchronously \\
   Bindings                    & actor addresses \\
   Services                    & transport with SOAP support \\
   \hline
  \end{tabular}
  \caption{Mappings between AB-WSCL and WSDL}
  \label{Table.1}
\end{table}
\end{center}

\begin{center}
\begin{table}

  \begin{tabular}{@{}ll@{}}
   \hline
   WS-BPEL\hspace{1cm}  & AB-WSCL \\
   \hline
   Process                          & WSO actor \\
   Atomic Activity                  & AA actor \\
   Sequence Composition             & causalities of messages sending among actors \\
   Parallel Composition             & parallel is intrinsical \\
   Choice Composition               & condition of method \\
   Loop composition                 & condition of method \\
   Partner Links                    & through a relation of the partner WS actor and the WSC actor \\
   Variables                        & local states and local variables \\
   \hline
  \end{tabular}
  \caption{Mappings between AB-WSCL and WS-BPEL}
  \label{Table.2}
\end{table}
\end{center}

\begin{center}
\begin{table}
  \begin{tabular}{@{}ll@{}}
   \hline
   WS-CDL\hspace{1cm}  & AB-WSCL \\
   \hline
   Choreography               & WSC actor \\
   Participant                & role definition \\
   Variables                  & local states and local variables \\
   Interaction                & linking two partner WS actors \\
   Expressions                & local variable assignments \\
   \hline
  \end{tabular}
  \caption{Mappings between AB-WSCL and WS-CDL}
  \label{Table.3}
\end{table}
\end{center}

\subsection{Advantages of AB-WSCL}

Through the above examinations of AB-WSCL, we can conclude that it has several advantages:

Firstly, AB-WSCL is an actor-based language that provides a new approach to unifying WSO and WSC. Though it is not the first work to unifying WSO and WSC, the actor-based approach provides more natural relationships among WSO, WS, and WSC under the environment of cross-organizational business integration.

Secondly, based on rewriting semantics for actors, we give AB-WSCL a strictly formal semantics between WSO and its interface WS, WS and its partner WS, WSO and its partner WSO.

Thirdly, AB-WSCL is quite simple but powerful. One aspect is that it can serve as a basis of several tools for WS Composition, such as modeling tools, simulation tools, verification tools and even WS Composition runtime systems. The other aspect is that we can easily translate from AB-WSCL into WS specifications or vice versa, and these translations make AB-WSCL a more useful work, but not only a work on paper.

Finally, the components of AB-WSCL are all actors, such as AAs, WSO actors, WSC actors and WS actors. This actor-constructed system can capture the concurrent nature of WS and WS Composition architecture because of the intrinsic concurrency that actors have.

\subsection{Future Works}\label{Future Works}

I believe that, In the future, explorations is needed of the enrichments of AB-WSCL, including providing pre-defined AA types to solve the requirements of WS Composition or even those of some domains of business process integrations, providing more states of AAs, WSO, WS, WSC for requirements of monitoring and controlling, etc.

A further aspect is about development of WS Composition system based on works of this paper, including modeling tools, simulation tools, verification tools, runtime systems, and also translating tools between AB-WSCL and WS specifications.

\newpage

\appendix
\section{XML-Based Web Service Specifications for Buying Books Example}\label{XMLDescription}

In Fig.\ref{Fig.5}, the user agent business process being modeled as UserAgent WSO described by WS-BPEL is described in following.

-------------------------------------------------------------------------------

$\langle$process name="UserAgent"

\quad targetNamespace="http://example.wscs.com/2011/ws-bp/useragent"

\quad xmlns="http://docs.oasis-open.org/wsbpel/2.0/process/executable"

\quad xmlns:lns="http://example.wscs.com/2011/wsdl/UserAgent.wsdl"

\quad xmlns:bns="http://example.wscs.com/2011/wsdl/BookStore.wsdl"$\rangle$

\quad $\langle$documentation xml:lang="EN"$\rangle$

\quad\quad This document describes the UserAgent process.

\quad $\langle$/documentation$\rangle$

\quad $\langle$partnerLinks$\rangle$

\quad\quad $\langle$partnerLink name="UserAndUserAgent"

\quad\quad\quad partnerLinkType="lns:UserAnduserAgentLT" myRole="userAgent"/$\rangle$

\quad\quad $\langle$partnerLink name="UserAgentAndBookStore"

\quad\quad\quad partnerLinkType="lns:UserAgentAndBookStoreLT"

\quad\quad\quad myRole="user" partnerRole="seller"/$\rangle$

\quad $\langle$/partnerLinks$\rangle$

\quad $\langle$variables$\rangle$

\quad\quad $\langle$variable name="RequestListofBooks" messageType="lns:requestListofBooks"/$\rangle$

\quad\quad $\langle$variable name="RequestListofBooksResponse" messageType="lns:requestListofBooksResponse"/$\rangle$

\quad\quad $\langle$variable name="ReceiveListofBooks" messageType="lns:receiveListofBooks"/$\rangle$

\quad\quad $\langle$variable name="ReceiveListofBooksResponse" messageType="lns:receiveListofBooksResponse"/$\rangle$

\quad\quad $\langle$variable name="SelectListofBooks"  messageType="lns:selectListofBooks"/$\rangle$

\quad\quad $\langle$variable name="SelectListofBooksResponse"  messageType="lns:selectListofBooksResponse"/$\rangle$

\quad\quad $\langle$variable name="ReceivePrice" messageType="lns:receivePrice"/$\rangle$

\quad\quad $\langle$variable name="ReceivePriceResponse" messageType="lns:receivePriceResponse"/$\rangle$

\quad\quad $\langle$variable name="Pays" messageType="lns:pays"/$\rangle$

\quad\quad $\langle$variable name="PaysResponse" messageType="lns:paysResponse"/$\rangle$

\quad $\langle$/variables$\rangle$

\quad $\langle$sequence$\rangle$

\quad\quad $\langle$receive partnerLink="UserAndUserAgent"

\quad\quad\quad portType="lns:userAgent4userInterface"

\quad\quad\quad operation="opRequestListofBooks" variable="RequestListofBooks"

\quad\quad\quad createInstance="yes"$\rangle$

\quad\quad $\langle$/receive$\rangle$

\quad\quad $\langle$invoke partnerLink="UserAgentAndBookStore"

\quad\quad\quad portType="bns:bookStore4userAgentInterface"

\quad\quad\quad operation="opRequestListofBooks" inputVariable="RequestListofBooks"

\quad\quad\quad outputVariable="RequestListofBooksResponse"$\rangle$

\quad\quad $\langle$/invoke$\rangle$

\quad\quad $\langle$receive partnerLink="UserAgentAndBookStore"

\quad\quad\quad portType="lns:userAgent4BookStoreInterface"

\quad\quad\quad operation="opReceiveListofBooks" variable="ReceiveListofBooks"$\rangle$

\quad\quad $\langle$/receive$\rangle$

\quad\quad $\langle$reply partnerLink="UserAgentAndBookStore"

\quad\quad\quad portType="lns:userAgent4BookStoreInterface"

\quad\quad\quad operation="opReceiveListofBooks" variable="ReceiveListofBooksResponse"$\rangle$

\quad\quad $\langle$/reply$\rangle$

\quad\quad $\langle$!--send the received book list to the user--$\rangle$

\quad\quad $\langle$receive partnerLink="UserAndUserAgent"

\quad\quad\quad portType="lns:userAgent4userInterface"

\quad\quad\quad operation="opSelectListofBooks" variable="SelectListofBooks"$\rangle$

\quad\quad $\langle$/receive$\rangle$

\quad\quad $\langle$reply partnerLink="UserAndUserAgent"

\quad\quad\quad portType="lns:userAgent4userInterface"

\quad\quad\quad operation="opSelectListofBooks" variable="SelectListofBooksResponse"$\rangle$

\quad\quad $\langle$/reply$\rangle$

\quad\quad $\langle$invoke partnerLink="UserAgentAndBookStore"

\quad\quad\quad portType="bns:bookStore4userAgentInterface"

\quad\quad\quad operation="opSelectListofBooks" inputVariable="SelectListofBooks"

\quad\quad\quad outputVariable="SelectListofBooksResponse"$\rangle$

\quad\quad $\langle$/invoke$\rangle$

\quad\quad $\langle$receive partnerLink="UserAgentAndBookStore"

\quad\quad\quad portType="lns:userAgent4BookStoreInterface"

\quad\quad\quad operation="opReceivePrice" variable="ReceivePrice"$\rangle$

\quad\quad $\langle$/receive$\rangle$

\quad\quad $\langle$reply partnerLink="UserAgentAndBookStore"

\quad\quad\quad portType="lns:userAgent4BookStoreInterface"

\quad\quad\quad operation="opReceivePrice" variable="ReceivePriceResponse"$\rangle$

\quad\quad $\langle$/reply$\rangle$

\quad\quad $\langle$!--send the price to the user and get pays from the user--$\rangle$

\quad\quad $\langle$invoke partnerLink="UserAgentAndBookStore"

\quad\quad\quad portType="bns:bookStore4userAgentInterface"

\quad\quad\quad operation="opPays" inputVariable="Pays" outputVariable="PaysResponse"$\rangle$

\quad\quad $\langle$/invoke$\rangle$

\quad\quad $\langle$reply partnerLink="UserAndUserAgent"

\quad\quad\quad portType="lns:userAgent4userInterface"

\quad\quad\quad operation="opRequestListofBooks" variable="PaysResponse"$\rangle$

\quad\quad $\langle$/reply$\rangle$

\quad $\langle$/sequence$\rangle$

$\langle$/process$\rangle$

-------------------------------------------------------------------------------

The interface WS for UserAgent WSO being called UserAgent WS described by WSDL is as following.

-------------------------------------------------------------------------------

$\langle$?xml version="1.0" encoding="utf-8"?$\rangle$

$\langle$description

\quad xmlns="http://www.w3.org/2004/08/wsdl"

\quad targetNamespace= "http://example.wscs.com/2011/wsdl/UserAgent.wsdl"

\quad\quad\quad xmlns:plnk="http://docs.oasis-open.org/wsbpel/2.0/plnktype"

\quad xmlns:tns= "http://example.wscs.com/2011/wsdl/UserAgent.wsdl"

\quad xmlns:ghns = "http://example.wscs.com/2011/schemas/UserAgent.xsd"

\quad xmlns:bsns = "http://example.wscs.com/2011/wsdl/BookStore.wsdl"

\quad xmlns:wsoap= "http://www.w3.org/2004/08/wsdl/soap12"

\quad xmlns:soap="http://www.w3.org/2003/05/soap-envelope"$\rangle$

\quad $\langle$documentation$\rangle$

\quad \quad This document describes the userAgent Web service.

\quad $\langle$/documentation$\rangle$

\quad$\langle$types$\rangle$

\quad\quad$\langle$xs:schema

\quad\quad\quad xmlns:xs="http://www.w3.org/2001/XMLSchema"

\quad\quad\quad targetNamespace="http://example.wscs.com/2011/schemas/UserAgent.xsd"

\quad\quad\quad xmlns="http://example.wscs.com/2011/schemas/UserAgent.xsd"$\rangle$

\quad\quad\quad $\langle$xs:element name="requestListofBooks" type="tRequestListofBooks"/$\rangle$

\quad\quad\quad $\langle$xs:complexType name="tRequestListofBooks"/$\rangle$

\quad\quad\quad $\langle$xs:element name="requestListofBooksReponse"

\quad\quad\quad\quad type="tRequestListofBooksResponse"/$\rangle$

\quad\quad\quad $\langle$xs:complexType name="tRequestListofBooksResponse"/$\rangle$

\quad\quad\quad $\langle$xs:element name="receiveListofBooks" type="tReceiveListofBooks"/$\rangle$

\quad\quad\quad $\langle$xs:complexType name="tReceiveListofBooks"/$\rangle$

\quad\quad\quad $\langle$xs:element name="receiveListofBooksResponse"

\quad\quad\quad\quad type="tReceiveListofBooksResponse"/$\rangle$

\quad\quad\quad $\langle$xs:complexType name="tReceiveListofBooksResponse"/$\rangle$

\quad\quad\quad $\langle$xs:element name="selectListofBooks" type="tSelectListofBooks"/$\rangle$

\quad\quad\quad $\langle$xs:complexType name="tSelectListofBooks"/$\rangle$

\quad\quad\quad $\langle$xs:element name="selectListofBooksResponse"

\quad\quad\quad\quad type="tSelectListofBooksResponse"/$\rangle$

\quad\quad\quad $\langle$xs:complexType name="tSelectListofBooksResponse"/$\rangle$

\quad\quad\quad $\langle$xs:element name="receivePrice" type="xs:float"/$\rangle$

\quad\quad\quad $\langle$xs:element name="receivePriceResponse" type="tReceivePriceResponse"/$\rangle$

\quad\quad\quad $\langle$xs:complexType name="tReceivePriceResponse"/$\rangle$

\quad\quad\quad $\langle$xs:element name="pays" type="tPays"/$\rangle$

\quad\quad\quad $\langle$xs:complexType name="tPays"/$\rangle$

\quad\quad\quad $\langle$xs:element name="paysResponse" type="tPaysResponse"/$\rangle$

\quad\quad\quad $\langle$xs:complexType name="tPaysResponse"/$\rangle$

\quad\quad $\langle$/xs:schema$\rangle$

\quad $\langle$/types$\rangle$

\quad $\langle$interface name = "UserAgent4UserInterface"$\rangle$

\quad\quad$\langle$operation name="opRequestListofBooks"$\rangle$

\quad\quad\quad $\langle$input messageLabel="InOpRequestListofBooks"

\quad\quad\quad\quad element="ghns:requestListofBooks" /$\rangle$

\quad\quad\quad $\langle$output messageLabel="OutOpRequestListofBooks"

\quad\quad\quad\quad element="ghns:requestListofBooksReponse" /$\rangle$

\quad\quad $\langle$/operation$\rangle$

\quad\quad $\langle$operation name="opSelectListofBooks"$\rangle$

\quad\quad\quad $\langle$input messageLabel="InOpSelectListofBooks"

\quad\quad\quad\quad element="ghns:selectListofBooks" /$\rangle$

\quad\quad\quad $\langle$output messageLabel="OutOpSelectListofBooks"

\quad\quad\quad\quad element="ghns:selectListofBooksResponse" /$\rangle$

\quad\quad $\langle$/operation$\rangle$

\quad $\langle$/interface$\rangle$

\quad $\langle$interface name = "UserAgent4BookStoreInterface"$\rangle$

\quad\quad $\langle$operation name="opReceiveListofBooks"$\rangle$

\quad\quad\quad $\langle$input messageLabel="InOpReceiveListofBooks"

\quad\quad\quad\quad element="ghns:receiveListofBooks" /$\rangle$

\quad\quad\quad $\langle$output messageLabel="OutOpReceiveListofBooks"

\quad\quad\quad\quad element="ghns:receiveListofBooksResponse" /$\rangle$

\quad\quad $\langle$/operation$\rangle$

\quad\quad $\langle$operation name="opReceivePrice"$\rangle$

\quad\quad\quad $\langle$input messageLabel="InOpReceivePrice"

\quad\quad\quad\quad element="ghns:receivePrice" /$\rangle$

\quad\quad\quad $\langle$output messageLabel="OutOpReceivePrice"

\quad\quad\quad\quad element="ghns:receivePriceResponse" /$\rangle$

\quad\quad $\langle$/operation$\rangle$

\quad $\langle$/interface$\rangle$

\quad $\langle$plnk:partnerLinkType name="UserAndUserAgentLT"$\rangle$

\quad\quad $\langle$plnk:role name="UserAgent"

\quad\quad\quad portType="tns:UserAgent4UserInterface" /$\rangle$

\quad $\langle$/plnk:partnerLinkType$\rangle$

\quad $\langle$plnk:partnerLinkType name="UserAgentAndBookStoreLT"$\rangle$

\quad\quad $\langle$plnk:role name="user"

\quad\quad\quad portType="tns:UserAgent4BookStoreInterface" /$\rangle$

\quad\quad $\langle$plnk:role name="seller"

\quad\quad\quad portType="bsns:BookStore4UserAgentInterface" /$\rangle$

\quad $\langle$/plnk:partnerLinkType$\rangle$

$\langle$/description$\rangle$

-------------------------------------------------------------------------------

In the buying books example, the WSC between user agent and bookstore (exactly UserAgentWS and BookStoreWS) called BuyingBookWSC being described by WS-CDL is following.

-------------------------------------------------------------------------------

$\langle$?xml version="1.0" encoding="UTF-8"?$\rangle$

$\langle$package xmlns="http://www.w3.org/2005/10/cdl"

\quad xmlns:cdl="http://www.w3.org/2005/10/cdl"

\quad xmlns:xsi="http://www.w3.org/2001/XMLSchema-instance"

\quad xmlns:xsd="http://www.w3.org/2001/XMLSchema"

\quad xmlns:bans="http://example.wscs.com/2011/wsdl/UserAgent.wsdl"

\quad xmlns:bsns="http://example.wscs.com/2011/wsdl/BookStore.wsdl"

\quad xmlns:tns="http://example.wscs.com/2011/cdl/BuyingBookWSC"

\quad targetNamespace="http://example.wscs.com/2011/cdl/BuyingBookWSC"

\quad name="BuyingBookWSC"

\quad version="1.0"$\rangle$

\quad $\langle$informationType name="requestListofBooksType" type="bsns:tRequestListofBooks"/$\rangle$

\quad $\langle$informationType name="requestListofBooksResponseType"

\quad\quad type="bsns:tRequestListofBooksResponse"/$\rangle$

\quad $\langle$informationType name="listofBooksType" type="bsns:tListofBooks"/$\rangle$

\quad $\langle$informationType name="listofBooksResponseType"

\quad\quad type="bsns:tListofBooksResponse"/$\rangle$

\quad $\langle$informationType name="selectListofBooksType"

\quad\quad type="bsns:tSelectListofBooks"/$\rangle$

\quad $\langle$informationType name="selectListofBooksResponseType"

\quad\quad type="bsns:tSelectListofBooksResponse"/$\rangle$

\quad $\langle$informationType name="priceType" type="bsns:tPrice"/$\rangle$

\quad $\langle$informationType name="priceResponseType" type="bsns:tPriceResponse"/$\rangle$

\quad $\langle$informationType name="paysType" type="bsns:tPays"/$\rangle$

\quad $\langle$informationType name="paysResponseType" type="bsns:tPaysResponse"/$\rangle$

\quad $\langle$roleType name="UserAgent"$\rangle$

\quad\quad $\langle$behavior name="UserAgent4BookStore" interface="bans:BuyAgent4BookStoreInterface"/$\rangle$

\quad $\langle$/roleType$\rangle$

\quad $\langle$roleType name="BookStore"$\rangle$

\quad\quad $\langle$behavior name="BookStore4userAgent" interface="rns:BookStore4userAgentInterface"/$\rangle$

\quad $\langle$/roleType$\rangle$

\quad $\langle$relationshipType name="UserAgentAndBookStoreRelationship"$\rangle$

\quad\quad $\langle$roleType typeRef="tns:user" behavior="UserAgent4BookStore"/$\rangle$

\quad\quad $\langle$roleType typeRef="tns:seller" behavior="BookStore4userAgent"/$\rangle$

\quad $\langle$/relationshipType$\rangle$

\quad $\langle$choreography name="BuyingBookWSC"$\rangle$

\quad\quad $\langle$relationship type="tns:UserAgentAndBookStoreRelationship"/$\rangle$

\quad\quad $\langle$variableDefinitions$\rangle$

\quad\quad\quad $\langle$variable name="requestListofBooks" informationType="tns:requestListofBooksType"/$\rangle$

\quad\quad\quad $\langle$variable name="requestListofBooksResponse"

\quad\quad\quad\quad informationType="tns:requestListofBooksResponseType"/$\rangle$

\quad\quad\quad $\langle$variable name="listofBooks" informationType="tns:listofBooksType"/$\rangle$

\quad\quad\quad $\langle$variable name="listofBooksResponse" informationType="tns:listofBooksResponseType"/$\rangle$

\quad\quad\quad $\langle$variable name="selectListofBooks" informationType="tns:selectListofBooksType"/$\rangle$

\quad\quad\quad $\langle$variable name="selectListofBooksResponse"

\quad\quad\quad\quad informationType="tns:selectListofBooksResponseType"/$\rangle$

\quad\quad\quad $\langle$variable name="price" informationType="tns:priceType"/$\rangle$

\quad\quad\quad $\langle$variable name="priceResponse" informationType="tns:priceResponseType"/$\rangle$

\quad\quad\quad $\langle$variable name="pays" informationType="tns:paysType"/$\rangle$

\quad\quad\quad $\langle$variable name="paysResponse" informationType="tns:paysResponseType"/$\rangle$

\quad\quad $\langle$/variableDefinitions$\rangle$

\quad\quad $\langle$sequence$\rangle$

\quad\quad\quad $\langle$interaction name="InteractionBetweenBAandBS1"$\rangle$

\quad\quad\quad\quad $\langle$participate relationshipType="tns:UserAgentAndBookStoreRelationship"

\quad\quad\quad\quad\quad fromRoleTypeRef="tns:user" toRoleTypeRef="tns:seller"/$\rangle$

\quad\quad\quad\quad $\langle$exchange name="requestListofBooks"

\quad\quad\quad\quad\quad informationType="tns:requestListofBooksType" action="request"$\rangle$

\quad\quad\quad\quad\quad $\langle$send variable="cdl:getVariable('tns:requestListofBooks','','')"/$\rangle$

\quad\quad\quad\quad\quad $\langle$receive variable="cdl:getVariable('tns:requestListofBooks','','')"/$\rangle$

\quad\quad\quad\quad $\langle$/exchange$\rangle$

\quad\quad\quad\quad $\langle$exchange name="requestListofBooksResponse"

\quad\quad\quad\quad\quad informationType="requestListofBooksResponseType" action="respond"$\rangle$

\quad\quad\quad\quad\quad $\langle$send variable="cdl:getVariable('tns:requestListofBooksResponse','','')"/$\rangle$

\quad\quad\quad\quad\quad $\langle$receive variable="cdl:getVariable('tns:requestListofBooksResponse','','')"/$\rangle$

\quad\quad\quad\quad $\langle$/exchange$\rangle$

\quad\quad\quad $\langle$/interaction$\rangle$

\quad\quad\quad $\langle$interaction name="InteractionBetweenBAandBS2"$\rangle$

\quad\quad\quad\quad $\langle$participate relationshipType="tns:UserAgentAndBookStoreRelationship"

\quad\quad\quad\quad\quad fromRoleTypeRef="tns:seller" toRoleTypeRef="tns:user"/$\rangle$

\quad\quad\quad\quad $\langle$exchange name="sendListofBooks"

\quad\quad\quad\quad\quad informationType="tns:listofBooksType" action="request"$\rangle$

\quad\quad\quad\quad\quad $\langle$send variable="cdl:getVariable('tns:listofBooks','','')"/$\rangle$

\quad\quad\quad\quad\quad $\langle$receive variable="cdl:getVariable('tns:listofBooks','','')"/$\rangle$

\quad\quad\quad\quad $\langle$/exchange$\rangle$

\quad\quad\quad\quad $\langle$exchange name="sendListofBooksResponse"

\quad\quad\quad\quad\quad informationType="listofBooksResponseType" action="respond"$\rangle$

\quad\quad\quad\quad\quad $\langle$send variable="cdl:getVariable('tns:listofBooksResponse','','')"/$\rangle$

\quad\quad\quad\quad\quad $\langle$receive variable="cdl:getVariable('tns:listofBooksResponse','','')"/$\rangle$

\quad\quad\quad\quad $\langle$/exchange$\rangle$

\quad\quad\quad $\langle$/interaction$\rangle$

\quad\quad\quad $\langle$interaction name="InteractionBetweenBAandBS3"$\rangle$

\quad\quad\quad\quad $\langle$participate relationshipType="tns:UserAgentAndBookStoreRelationship"

\quad\quad\quad\quad\quad fromRoleTypeRef="tns:user" toRoleTypeRef="tns:seller"/$\rangle$

\quad\quad\quad\quad $\langle$exchange name="selectListofBooks"

\quad\quad\quad\quad\quad informationType="tns:selectListofBooksType" action="request"$\rangle$

\quad\quad\quad\quad\quad $\langle$send variable="cdl:getVariable('tns:selectListofBooks','','')"/$\rangle$

\quad\quad\quad\quad\quad $\langle$receive variable="cdl:getVariable('tns:selectListofBooks','','')"/$\rangle$

\quad\quad\quad\quad $\langle$/exchange$\rangle$

\quad\quad\quad\quad $\langle$exchange name="selectListofBooksResponse"

\quad\quad\quad\quad\quad informationType="selectListofBooksResponseType" action="respond"$\rangle$

\quad\quad\quad\quad\quad $\langle$send variable="cdl:getVariable('tns:selectListofBooksResponse','','')"/$\rangle$

\quad\quad\quad\quad\quad $\langle$receive variable="cdl:getVariable('tns:selectListofBooksResponse','','')"/$\rangle$

\quad\quad\quad\quad $\langle$/exchange$\rangle$

\quad\quad\quad $\langle$/interaction$\rangle$

\quad\quad\quad $\langle$interaction name="InteractionBetweenBAandBS4"$\rangle$

\quad\quad\quad\quad $\langle$participate relationshipType="tns:UserAgentAndBookStoreRelationship"

\quad\quad\quad\quad\quad fromRoleTypeRef="tns:seller" toRoleTypeRef="tns:user"/$\rangle$

\quad\quad\quad\quad $\langle$exchange name="sendPrice"

\quad\quad\quad\quad\quad informationType="tns:priceType" action="request"$\rangle$

\quad\quad\quad\quad\quad $\langle$send variable="cdl:getVariable('tns:price','','')"/$\rangle$

\quad\quad\quad\quad\quad $\langle$receive variable="cdl:getVariable('tns:price','','')"/$\rangle$

\quad\quad\quad\quad $\langle$/exchange$\rangle$

\quad\quad\quad\quad $\langle$exchange name="sendPriceResponse"

\quad\quad\quad\quad\quad informationType="priceResponseType" action="respond"$\rangle$

\quad\quad\quad\quad\quad $\langle$send variable="cdl:getVariable('tns:priceResponse','','')"/$\rangle$

\quad\quad\quad\quad\quad $\langle$receive variable="cdl:getVariable('tns:priceResponse','','')"/$\rangle$

\quad\quad\quad\quad $\langle$/exchange$\rangle$

\quad\quad\quad $\langle$/interaction$\rangle$

\quad\quad\quad $\langle$interaction name="InteractionBetweenBAandBS5"$\rangle$

\quad\quad\quad\quad $\langle$participate relationshipType="tns:UserAgentAndBookStoreRelationship"

\quad\quad\quad\quad\quad fromRoleTypeRef="tns:user" toRoleTypeRef="tns:seller"/$\rangle$

\quad\quad\quad\quad $\langle$exchange name="pays"

\quad\quad\quad\quad\quad informationType="tns:paysType" action="request"$\rangle$

\quad\quad\quad\quad\quad $\langle$send variable="cdl:getVariable('tns:pays','','')"/$\rangle$

\quad\quad\quad\quad\quad $\langle$receive variable="cdl:getVariable('tns:pays','','')"/$\rangle$

\quad\quad\quad\quad $\langle$/exchange$\rangle$

\quad\quad\quad\quad $\langle$exchange name="paysResponse"

\quad\quad\quad\quad\quad informationType="paysResponseType" action="respond"$\rangle$

\quad\quad\quad\quad\quad $\langle$send variable="cdl:getVariable('tns:paysResponse','','')"/$\rangle$

\quad\quad\quad\quad\quad $\langle$receive variable="cdl:getVariable('tns:paysResponse','','')"/$\rangle$

\quad\quad\quad\quad $\langle$/exchange$\rangle$

\quad\quad\quad $\langle$/interaction$\rangle$

\quad\quad $\langle$/sequence$\rangle$

\quad $\langle$/choreography$\rangle$

$\langle$/package$\rangle$

-------------------------------------------------------------------------------

\label{lastpage}

\end{document}